\documentclass[11pt,twoside]{amsart}
\usepackage[utf8]{inputenc}

\usepackage{geometry}
\usepackage[normalem]{ulem}
\usepackage{enumerate}
\usepackage{setspace}
\usepackage{amstext}
\usepackage{xcolor}
\usepackage{tikz}
\usepackage{multirow} 

\usepackage{url}
\usepackage[colorlinks,urlcolor=red,citecolor=blue,linkcolor=red]{hyperref}
\usepackage[round,comma]{natbib} 

\usepackage{amsmath}
\usepackage{amsthm}
\usepackage{amsfonts}
\usepackage{amssymb}
\usepackage{dsfont}
\usepackage{nicefrac}
\usepackage{booktabs}
\usepackage{mathrsfs}
\usepackage{bm}
\usepackage{mathtools}

\newcommand{\cB}{{\mathcal B}}
\newcommand{\cC}{{\mathcal C}}

\newcommand{\cF}{{\mathcal F}}

\newcommand{\crL}{{\mathscr L}}\newcommand{\cL}{{\mathcal L}}
\newcommand{\cM}{{\mathcal M}}

\newcommand{\cO}{{\mathcal O}}
\newcommand{\cP}{{\mathcal P}}
\newcommand{\cQ}{{\mathcal Q}}
\newcommand{\cR}{{\mathcal R}}

\newcommand{\cT}{{\mathcal T}}

\newcommand{\cV}{{\mathcal V}}

\newcommand{\bbF}{\mathbb{F}}

\newcommand{\bbE}{\mathbb{E}}
\newcommand{\bbP}{\mathbb{P}}

\newcommand{\R}{\mathbb R} 

\newcommand{\Ind}{{\mathds 1}}

\newcommand{\restr}{\mathbf{\kern0.3ex%
 \vert\kern-0.3ex}\backprime\kern0.3ex}
\DeclarePairedDelimiter{\abs}{\lvert}{\rvert}
\DeclarePairedDelimiter{\norm}{\lVert}{\rVert}

\newtheorem{theorem}{Theorem}[section]
\newtheorem{lemma}[theorem]{Lemma}              
\newtheorem{proposition}[theorem]{Proposition}  
\theoremstyle{definition}
\newtheorem{example}[theorem]{Example} 
\newtheorem{definition}[theorem]{Definition} 
\newtheorem{remark}[theorem]{Remark} 
\newtheorem{assumption}[theorem]{Assumption} 

\usepackage[backgroundcolor=white,bordercolor=orange]{todonotes}

\usepackage[colorlinks,urlcolor=red,citecolor=blue,linkcolor=red]{hyperref}

\usepackage{subfiles} 

\newcounter{stepcounter} 
\newcommand{\step}[1][]{%
    \refstepcounter{stepcounter}
    \textit{Step \thestepcounter
    \ifx\\#1\\
    \else: #1
    \fi.}
}

\title[Propagation of carbon price shocks through the value chain]{Propagation of carbon price shocks through the value chain: the mean-field game of defaults}

\author{Zorana Grbac}
\address{LPSM, Université Paris Cité}
\email{grbac@math.univ-paris-diderot.fr}

\author{Simone Pavarana}
\address{Albert-Ludwigs-Universitaet Freiburg, Abteilung für Mathematische Stochastik}
\email{simone.pavarana@stochastik.uni-freiburg.de}

\author{Thorsten Schmidt}
\address{Albert-Ludwigs-Universitaet Freiburg, Abteilung für Mathematische Stochastik}
\email{thorsten.schmidt@stochastik.uni-freiburg.de}

\author{Peter Tankov}
\address{CREST, ENSAE, Institut Polytechnique de Paris}
\email{peter.tankov@ensae.fr}

\begin{document}

\begin{abstract} 
We introduce a new mean-field game framework to analyze the impact of carbon pricing in a multi-sector economy with defaultable firms. Each sector produces a homogeneous good, with its price endogenously determined through market clearing. Firms act as price takers and maximize profits by choosing an optimal allocation of inputs—including labor, emissions, and intermediate goods from other sectors—while interacting through the endogenous sectoral price. Firms also choose their default timing to maximize shareholder value.

Formally, we model the economy as an optimal stopping mean-field game within each sector. The resulting system of coupled mean-field games admits a linear programming formulation that characterizes Nash equilibria in terms of population measure flows. We prove the existence of a linear programming Nash equilibrium and establish uniqueness of the associated price system.

Numerical illustrations are presented for firms with constant elasticity of substitution (CES) production functions. In a stylized single-sector economy, carbon price shocks induce substitution between emissions and labor. In a three-sector economy, the manufacturing sector faces consumer demand and requires inputs from a brown sector, which can be increasingly replaced by green-sector goods as carbon prices rise. These experiments reveal that carbon price shocks can generate substantial spillover effects along the value chain, underscoring the importance of sectoral interdependencies in shaping effective decarbonization pathways.
\end{abstract}


\maketitle
\section{Introduction}\label{intro}

A large-scale shift toward low-carbon technologies is essential to limit global warming and mitigate the most severe impacts of climate change. As governments implement policies to reduce emissions and consumer sentiment drives demand shocks, green firms are expected to gain market share, while brown firms—those unable to decarbonize—may face revenue losses, asset stranding, and even default. These transition risks can propagate through the real economy and the financial system, and must therefore be closely monitored by banks and supervisory authorities \citep{basel2021climate,basel2022principles}.

The prevailing methodology for assessing transition risks relies on climate stress testing \citep{acharya2023climate}. A climate stress test evaluates the potential losses of financial institutions under specific transition scenarios using a combination of economic, financial, and physical models. The typical approach begins with an integrated assessment model that translates carbon price or demand shocks into a global macroeconomic trajectory and ends with a default model projecting firm-level distress \citep{ACPR2020pilotexercise}. While pragmatic, this multi-model setup suffers from important limitations: inconsistencies across models, a lack of two-way feedback between macro and micro levels, and limited treatment of uncertainty.

In this paper, we propose an alternative, integrated modeling framework to assess the impact of macroeconomic shocks—particularly carbon pricing—on transition risks and firm defaults. Our approach is based on the theory of mean-field games (MFGs), which provides the mathematical tractability needed to jointly capture endogenous default decisions under idiosyncratic shocks and realistic inter-sectoral input-output linkages within a unified model.

We develop a multi-sector equilibrium model of the productive economy aimed at studying how carbon price shocks propagate through the value chain. Each sector comprises a continuum of firms producing a homogeneous good. Production requires inputs from other sectors, labor, and carbon emissions, and is described by a sector-specific production function. Sectors face endogenous demand from other sectors and exogenous consumer demand, and prices are determined endogenously via market clearing. Firms choose an optimal allocation of inputs—including labor, emissions, and goods from other sectors—to maximize profits. While production technologies are identical within each sector, firms are subject to idiosyncratic labor costs and interact through sectoral prices both within and across sectors.

Firms also choose their optimal default timing to maximize shareholder value, resulting in exits at uncertain future dates. Since carbon emissions are essential to production, carbon price shocks raise marginal costs and may trigger firm default. A cascade of defaults in a carbon-sensitive sector can drive up the price of its output, thereby increasing costs in downstream sectors and amplifying default risk along the value chain.

We model the economy as a coupled system of optimal stopping mean-field games—one for each sector. Following \cite{bouveret2020mean} and \cite{Dumitrescu2021LP}, we solve this system using a linear programming formulation, which characterizes Nash equilibria in terms of population measure flows. We prove the existence of a linear programming Nash equilibrium and establish uniqueness of the associated price system. To this end, we introduce an auxiliary Minimax problem and show that any saddle point corresponds to a MFG Nash equilibrium. This formulation enables us to prove existence via the existence of saddle points. The Minimax problem admits a natural interpretation as a social planner problem, in which the planner controls the distribution of agents and selects the price system to maximize aggregate surplus.

Finally, we provide numerical illustrations using CES production functions and Cox – Ingersoll – Ross (CIR) processes for labor cost dynamics. We explore two scenarios: one with a single sector exposed to carbon price shocks, and another with three interconnected sectors (brown, green, and manufacturing) linked through a specified input–output network. Our simulations show that in the directly affected sector, capacity declines significantly in response to carbon pricing. Moreover, the shock propagates through the value chain, generating spillover effects. However, even large shocks rarely lead to abrupt waves of default, as firms optimally anticipate the rise in emission costs and, after an initial clearing, exit the market gradually.

These findings support the view that an early-announced and stable policy framework facilitates smoother asset value adjustments and mitigates systemic risk. 

\subsection{Structure of the paper}

The remainder of the paper is organized as follows. The rest of this introduction reviews the relevant literature. Section~\ref{sec:1} introduces the profit maximization problem of a single firm in a static setting. We prove the existence of profit-maximizing input configurations and provide explicit solutions in the case of CES production functions. Section~\ref{sec:2} extends the analysis to a multi-sector economy with a static price formation mechanism. We study the interaction between prices and input choices from a competitive equilibrium perspective and prove the existence of a competitive (Nash) equilibrium. Section~\ref{sec:3} develops a dynamic market model, where each sector is represented by a mean-field firm optimizing both production and default timing. Using the linear programming approach introduced in \citet{bouveret2020mean} and \citet{Dumitrescu2021LP}, we reformulate the optimal stopping problem as a linear program. We then establish the existence of a mean-field game Nash equilibrium and prove the uniqueness of the associated price system. Section~\ref{sec:4} presents numerical illustrations, and Section~\ref{sec:conclusion} concludes.

\subsection{Related literature}

We briefly review the literature relevant to our work.
Comprehensive overviews of the financial implications of climate change are provided by \cite{campiglio2018climate} and \cite{krueger2020importance}. The former emphasizes the role of central banks and financial regulators in facilitating the low-carbon transition, particularly through carbon pricing. The latter focuses on institutional investors and the influence of climate risks — especially regulatory risks — on portfolio decisions.

This focus on risk has spurred stress-testing methodologies to assess financial vulnerabilities from regulatory shocks. \cite{battiston2017climate} propose a network-based approach to capture financial interdependencies, while \cite{bouchet2020credit} use a bottom-up framework to evaluate corporate credit risk under carbon pricing scenarios. Several studies build on the Merton model \citep{Merton1974PricingCorporateDebt}, including \cite{Capasso2020climate} and \cite{Reinders2023finance}, who empirically show that firms with higher carbon footprints face higher perceived default risk, reinforcing the need to incorporate carbon exposure into credit risk assessments.

Further structural models include \cite{Wang2024carbon} and \cite{loschenbrand2024credit}. The former uses China’s accession to the Paris Agreement as an exogenous shock, finding that rising carbon risk can reduce defaults among high-emission firms — a possible regulatory gain. The latter combines financial and emissions data for over 2.5 million borrowers across major European banks. Their stress tests suggest that while high carbon prices pressure certain firms, systemic default risk in the Euro area remains limited.

Building on an earlier work by \cite{agliardi2021pricing}, \cite{le2025corporate} introduce an endogenous default model with partial information to price defaultable bonds under transition risk, whose severity  depends on a latent scenario gradually revealed through the dynamics of the carbon price.

A separate strand of literature studies the diffusion of carbon pricing across sectors using Input-Output (I-O) and equilibrium models. I-O analysis, rooted in Leontief’s work, models intersectoral dependencies \citep{miller09}. \cite{adenot2022cascading} apply this framework to the global economy using WIOD data under different carbon price scenarios. Related studies on Chile include \cite{mardones2018environmental} and \cite{mardones2020economic}.

Despite their tractability, I-O models assume fixed input proportions and do not model price responses. Equilibrium models address these limitations by allowing input substitution and price formation, typically under perfect competition.

Partial and general equilibrium models have been widely used to analyze carbon pricing. \cite{neuhoff2019carbon} adopt a partial equilibrium approach to study carbon pass-through within an industry. \cite{Frankovic2024carbon} uses a multisector, multiregional computable general equilibrium (CGE) model to evaluate spillover effects of global carbon pricing on Germany and Europe (see also \cite{wei2023climate}).

To incorporate time dynamics and uncertainty, these models evolve into dynamic stochastic general equilibrium (DSGE) frameworks, enabling richer firm behavior such as entry and exit. \cite{Miao2005OptimalCapital} develops a single-sector DSGE model with firm-level technology shocks, anticipating ideas later formalized in mean-field games (MFGs) by \cite{lasry2007mean}.

In the context of transition risk, \cite{Bouveret2023} build a DSGE model incorporating GHG emission costs in both production and consumption, linking carbon costs to credit portfolios. Similarly, \cite{MATSUMURA2024} model Japan’s economy using a DSGE framework with I-O structure and investment networks.

DSGE models also offer normative insights. \cite{Golosov2014Optimal} derive a closed-form optimal carbon tax by solving a social planner’s problem with climate externalities, later extended by \cite{rezai2021optimal} to include global warming’s impact on utility and productivity growth.

While powerful, DSGE models become intractable with many heterogeneous agents. Mean-field game theory, introduced by \cite{lasry2007mean}, overcomes this curse of dimensionality by shifting the focus from individual agents to the ``mean field", which represents the aggregate behaviour of the population. MFGs arise as the asymptotic limit of stochastic differential games with many symmetric agents. Each agent controls a state variable influenced by idiosyncratic noise and the population distribution. In the limit, this leads to a coupled system of PDEs: a Hamilton-Jacobi-Bellman (HJB) equation for the value function and a Fokker-Planck (FP) equation for the distribution dynamics. For a comprehensive discussion, see \cite{Carmona1} and \cite{Carmona2}.

While early MFG studies focused on solving this PDE system, an alternative formulation by \cite{carmona2015forward} characterizes equilibria using forward-backward stochastic differential equations (FBSDEs).

Extensions to MFGs include common noise, major players, singular controls, and optimal stopping. The latter is particularly relevant for entry-exit problems, though solving the HJB-FP system becomes difficult due to the free boundary. To overcome this, \cite{bouveret2020mean}, \cite{Dumitrescu2021LP}, and \cite{dumitrescu2023linear} propose a novel reformulation of the optimal stopping problem as a linear program, where optimization is performed over population measure flows subject to a linear constraint imposed by the FP equation.

Mean-field games have found wide application in economics. For instance, \cite{Achdou2022} extend the Aiyagari–Bewley–Huggett model of income and wealth distribution to continuous time, using mean-field game theory to handle agent heterogeneity. In a related context, \cite{Achdou2023} model an economy composed of competing firms subject to idiosyncratic capital dynamics, focusing on stationary equilibria for the consumption problem under capital constraints. Further relevant applications of MFG theory in economics can be found in \cite{Alvarez2023} and \cite{Liang2024}.

Applications of MFG modeling to climate and energy topics—including electricity markets and green transitions—have gained increasing attention. Notable contributions include \cite{Gueant2010, Escribe2024, aid2021entry, Bassiere2024elect}. The latter two, in particular, address entry–exit problems in electricity markets involving conventional and renewable producers, employing a linear programming approach to characterize the MFG Nash equilibria. Their existence and uniqueness results rely on fixed-point arguments. In contrast, the present work also adopts the linear programming formulation but establishes the existence of equilibria via a Minimax approach.

\section{Firm-level profit maximization with carbon pricing}\label{sec:1}

In this section, we formulate and solve the profit maximization problem for a single firm operating in a fixed sector $i\in \{1,\dots,N\}$ in interaction with the other sectors. The production process is modeled by a function that combines inputs from all sectors, carbon emissions considered as input in the production process, and a local input such as capital or labor. The firm operates in a competitive market, optimizing its input decisions to maximize profits while treating prices as given. Under general assumptions on the production function and cost structure, we establish the existence of an optimal input choice. In the special case of a constant elasticity of substitution (CES) production function, we derive closed-form solutions for the firm's optimal input allocation and output level. Proofs are presented in Appendix \ref{app}. 

\subsection{Production functions}

Let \( q_{ij} \) denote the quantity of good \( j \) used by the firm in its production process. We define \( \mathbf{q}_i \coloneqq (q_{ij})_{j=1,\dots,N} \) as the vector of input quantities sourced by the firm from all sectors, including its own sector \( i \).

In addition to sectoral inputs, we introduce two specific production inputs: \( E_i \), referred to as \emph{emissions}, representing greenhouse gas (GHG) emissions resulting from the production process; and \( L_i \), referred to as \emph{labor}, encompassing both workforce and productive capital (e.g., machinery) employed in production.

By treating emissions as an input — or ``quasi-input" — to production, we follow the approach adopted in the environmental economics literature, notably in \cite{ebert2007environmental} and \cite{considine2006environment}. This formulation allows firms to substitute carbon-intensive technologies with labor - or capital - intensive alternatives, or with technologies relying on inputs from other sectors, in order to mitigate emission-related costs.

The production function \( F_i : \R^{N+2}_{\geq 0} \to \R_{\geq 0} \) maps the input quantities to the gross total output of the firm, which is given by
\begin{equation}\label{eq:output}
    F_i(\mathbf{q}_i, E_i, L_i), \quad i = 1, \dots, N.
\end{equation}
To simplify notation, we denote the full vector of inputs as \( \mathbf{x} \coloneqq (\mathbf{q}, E, L) \in \R^{N+2}_{\geq 0} \). We assume that the production function satisfies the following standard conditions.

\begin{assumption}[Production function]\label{ass:prod_function}
The production function $F_i$ is increasing in each argument, upper semi-continuous, and concave. Additionally, $F_i$ satisfies one of the following alternative conditions:
\begin{itemize}
\item[(a)] $F_i(\mathbf x) = o(\norm{\mathbf x})$ as $\norm{\mathbf x}\to \infty$. 
\item[(b)] $F_i$ is homogeneous of degree one, meaning that 
\begin{align*}
       F_i(\lambda\mathbf{x})=\lambda F_i(\mathbf{x}),
    \end{align*}
    for any $\lambda>0$ and $\mathbf x\in \mathbb R^{N+2}_{\geq 0}$, and additionally, 
    \[
    F_i(\mathbf{q},E,L)  \leq c_i L \quad \text{for some constant $c_i<\infty$}.
    \]
\end{itemize}
\end{assumption}

Condition (a) implies that the production function exhibits sublinear growth in total input usage, capturing inefficiencies that arise at high input levels and leading to diminishing marginal gains. This extends the classical notion of \emph{decreasing returns to scale}, typically characterized by homogeneity of degree less than one.

Condition (b), in contrast, corresponds to \emph{constant returns to scale}, where output scales proportionally with all inputs. The additional upper bound relative to labor ensures that labor remains an essential input, preventing it from being completely substituted by other production factors.

For a general discussion of production functions and the economic interpretation of their properties, we refer to standard microeconomic textbooks such as \cite{Varian92} and \cite{Mas-colell95}.

\begin{example}[CES Functions]
The CES (Constant Elasticity of Substitution) production functions form a broad and widely used class of production technologies. In our framework, the CES production function takes the form
\begin{equation}\label{eq:CES_emission}
F_i(\mathbf{q}_i,E_i,L_i) = A_i \left( \sum_{j=1}^N \alpha_{ij} q_{ij}^{-\rho_i} + \alpha_{iE} E_i^{-\rho_i} + \alpha_{iL} L_i^{-\rho_i} \right)^{-\frac{k_i}{\rho_i}}, \quad \rho_i > 0,\; k_i \in (0,1],
\end{equation}
where \( A_i \geq 0 \) denotes the total factor productivity, and the share parameters \( \alpha_{ij} \), \( \alpha_{iE} \), and \( \alpha_{iL} \) satisfy
\[
\sum_{j=1}^N \alpha_{ij} + \alpha_{iE} + \alpha_{iL} = 1.
\]
The substitution parameter \( \rho_i \) determines the elasticity of substitution via
\[
ES_i = \frac{1}{1 + \rho_i}.
\]
The parameter \( k_i \) governs the degree of homogeneity: for \( k_i = 1 \), the function is linearly homogeneous; for \( k_i < 1 \), it exhibits sublinear growth.

This class satisfies Assumption~\ref{ass:prod_function}: condition (a) holds for \( k_i < 1 \), and condition (b) for \( k_i = 1 \), with \( c_i = \alpha_{iL} \) ensuring a linear bound in labor. See \cite{Uzawa62, Arrow61} for further details.
\end{example}

\subsection{Firm's profit maximization problem}

We assume that the firm operates in a competitive market and behaves as a price taker. 

The total cost function $C_i:\R^{2N+3}_{\geq 0} \to \R_{\geq 0}$ is given by:
\begin{equation}\label{eq:cost}
C_i(\mathbf{P},P_E,\mathbf{q}_i,E_i,L_i) = \sum^N_{j=1} P_j q_{ij} + P_E E_i + W_i(L_i),
\end{equation}
where $\mathbf{P} = (P_1, \dots, P_N) \in \R^N_{\geq 0}$ denotes the vector of sectoral input prices, $P_E \in \R_{>0}$ is the carbon price, and $W_i:\R_{\geq 0} \to \R_{>0}$ is the labor cost function, assumed to satisfy the following conditions.

\begin{assumption}[Labor cost function]\label{ass:wage_function}
The labor cost function $W_i$ is strictly increasing, lower semi-continuous, and convex.
\end{assumption}

If the production function $F_i$ satisfies Assumption \ref{ass:prod_function}~(b), we shall impose an additional growth condition on $W_i$ to ensure well-posedness of the maximization problem.

\begin{assumption}[Asymptotic growth of labor cost function]\label{ass:wage_function2}
The labor cost function $W_i$ satisfies
\[
\lim_{L \to \infty} \frac{W_i(L)}{L} = +\infty.
\]
\end{assumption}

\begin{example}[Power labor cost functions]\label{ex:power-labor}
A class of labor cost functions satisfying both Assumption \ref{ass:wage_function} and Assumption \ref{ass:wage_function2} is given by:
\[
W_i(L) = \frac{L^{\eta_i}}{\eta_i},
\]
for some exponent $\eta_i > 1$. This specification implies that marginal labor costs increase with the scale of input, capturing increasing difficulty or expense in scaling up labor beyond a certain level — due, for example, to overtime pay, training needs, or capital constraints.
\end{example}

Given a price vector $\mathbf{P} \in \R^N_{\geq 0}$ and a carbon price $P_E \in \R_{>0}$, the firm's profit function is defined as total revenue minus total cost:
\begin{equation}\label{eq:profit}
\Pi_i(\mathbf{P},P_E,\mathbf{q}_i,E_i,L_i) \coloneqq P_i F_i(\mathbf{q}_i,E_i,L_i) - C_i(\mathbf{P},P_E,\mathbf{q}_i,E_i,L_i).
\end{equation}

The firm's profit maximization problem then consists of choosing an optimal input vector $(\mathbf{q}_i^*, E_i^*, L_i^*)$ that maximizes profits:
\begin{equation}\label{prob:profit_max}
\overline{\Pi}_i(\mathbf{P},P_E) \coloneqq \max_{\mathbf{q}_i, E_i, L_i} \Pi_i(\mathbf{P},P_E,\mathbf{q}_i,E_i,L_i).
\end{equation}

We now establish the existence and uniqueness of an optimal solution to the profit maximization problem, adapting classical results from microeconomic theory to our setting. 

\begin{proposition}\label{prop:profit_max}
Suppose the production and labor cost functions satisfy either Assumptions \ref{ass:prod_function} (a) and \ref{ass:wage_function}, or Assumptions \ref{ass:prod_function} (b), \ref{ass:wage_function} and \ref{ass:wage_function2}. Then, Problem \eqref{prob:profit_max} admits a maximiser $\mathbf{x}^*_i=(\mathbf{q}^*_i,E^*_i,L^*_i)\in\R^{N+2}_{\geq 0}$, for any price vector $\mathbf{P}\in\R^N_{>0}$ and any carbon price $P_E>0$. Moreover, if the production function $F_i$ is strictly concave and the labor cost function $W_i$ is strictly convex, the solution is unique.  
\end{proposition}

If the production function \( F_i \) satisfies Assumption~\ref{ass:prod_function}~(b), its homogeneity property allows the maximized profit to be expressed in terms of the convex conjugate of the labor cost function, denoted as \( W^*_i : \R \to \R \) and defined by
\begin{align}\label{convex conjugate W}
W^*_i(y) = \max_{L \geq 0} \left\{ L \cdot y - W_i(L) \right\}, \quad \text{for any } y \in \R.
\end{align}
Under Assumptions~\ref{ass:wage_function} and~\ref{ass:wage_function2}, the maximum in~\eqref{convex conjugate W} is attained, so that \( W^*_i \) is a proper convex function satisfying \( W^*_i(y) = 0 \) for all \( y \leq 0 \).

\begin{proposition}\label{prop:CRS}
Suppose the production and labor cost functions satisfy Assumptions \ref{ass:prod_function} (b), \ref{ass:wage_function}, and \ref{ass:wage_function2}. Then, the maximized profit from Problem~\eqref{prob:profit_max} satisfies
\[
\overline{\Pi}_i(\mathbf{P},P_E) = W^*_i\big(\widetilde{\Pi}_i(\mathbf{P},P_E)\big),
\]
where the function \(\widetilde{\Pi}_i\) is defined as the value of the following auxiliary maximization problem:
\begin{align}\label{pitilde}
\widetilde{\Pi}_i(\mathbf{P},P_E) = \max_{\tilde{\mathbf{q}}_i, \tilde{E}_i}\Bigg\{P_i F_i(\tilde{\mathbf{q}}_i, \tilde{E}_i,1) - \sum_{j=1}^N P_j \tilde{q}_{ij} - P_E \tilde{E}_i \Bigg\}.  
\end{align}  

Moreover, if \(W_i\) is strictly convex and differentiable, the optimal input vector \((\mathbf{q}^*_i, E^*_i, L^*_i)\) is given by
\[
L^*_i = (W'_i)^{-1}\big(\widetilde{\Pi}_i(\mathbf{P},P_E)\big),\quad q^*_{ij} = \tilde{q}^*_{ij} L^*_i\quad\text{for } j=1,\dots,N,\quad E^*_i = \tilde{E}^*_i L^*_i.
\]
\end{proposition}

The function \(\widetilde{\Pi}_i(\mathbf{P}, P_E)\) admits a natural economic interpretation: it represents the maximum profit per unit of labor, before deducting labor costs. For CES production technologies, this quantity can be computed explicitly.

\begin{proposition}[Explicit expression under CES technology]\label{ces.ex}
Assume the production function takes the CES form:
\[
F_i(\mathbf{q}_i, E_i, L_i) = A_i \left( \sum_{j=1}^N \alpha_{ij} q_{ij}^{-\rho_i} + \alpha_{iE} E_i^{-\rho_i} + \alpha_{iL} L_i^{-\rho_i} \right)^{-\frac{1}{\rho_i}},
\quad \text{with } \rho_i > 0.
\]
Then the optimal normalized input vector \((\tilde{\mathbf{q}}^*_i, \tilde{E}^*_i)\) is given as follows:
\begin{itemize}
    \item If
    \[
    \sum_{j=1}^N \left( \frac{P_j \alpha_{ij}^{1/\rho_i}}{A_i P_i} \right)^{\frac{\rho_i}{\rho_i+1}} 
    + \left( \frac{P_E \alpha_{iE}^{1/\rho_i}}{A_i P_i} \right)^{\frac{\rho_i}{\rho_i+1}} \leq 1,
    \]
    then the optimal production per unit of labor, \(\tilde{Q}^*_i \coloneqq F_i(\tilde{\mathbf{q}}^*_i, \tilde{E}^*_i, 1)\), is given by
    \[
    \tilde{Q}^*_i = A_i \alpha_{iL}^{-1/\rho_i} \left( 
    1 - \sum_{j=1}^N \left( \frac{P_j \alpha_{ij}^{1/\rho_i}}{A_i P_i} \right)^{\frac{\rho_i}{\rho_i+1}} 
    - \left( \frac{P_E \alpha_{iE}^{1/\rho_i}}{A_i P_i} \right)^{\frac{\rho_i}{\rho_i+1}} 
    \right)^{\frac{1}{\rho_i}}.
    \]
    The corresponding optimal inputs are
    \[
    \tilde{q}^*_{ij} = \tilde{Q}^*_i \left( \frac{P_j A_i^{\rho_i}}{P_i \alpha_{ij}} \right)^{-\frac{1}{\rho_i+1}}, 
    \qquad 
    \tilde{E}^*_i = \tilde{Q}^*_i \left( \frac{P_E A_i^{\rho_i}}{P_i \alpha_{iE}} \right)^{-\frac{1}{\rho_i+1}}.
    \]
    In this case, the maximum profit per unit of labor is
    \[
    \widetilde{\Pi}_i(\mathbf{P}, P_E) 
    = P_i A_i \alpha_{iL}^{-1/\rho_i} \left( 
    1 - \sum_{j=1}^N \left( \frac{P_j \alpha_{ij}^{1/\rho_i}}{A_i P_i} \right)^{\frac{\rho_i}{\rho_i+1}} 
    - \left( \frac{P_E \alpha_{iE}^{1/\rho_i}}{A_i P_i} \right)^{\frac{\rho_i}{\rho_i+1}} 
    \right)^{\frac{\rho_i+1}{\rho_i}}.
    \]
    
\item Otherwise, the optimum is attained at \(\tilde{q}^*_{ij} = 0\) for all \(j = 1, \dots, N\), \(\tilde{E}^*_i = 0\), and the resulting optimal production is \(\tilde{Q}^*_i = 0\).
\end{itemize}
\end{proposition}

Proposition \eqref{prop:CRS} can then be applied to explicitly compute the optimal profit and the corresponding input allocation, provided the labor cost function admits a known convex conjugate — for instance, the power function in Example \ref{ex:power-labor}.

\section{Price formation and competitive equilibrium in the static setting}\label{sec:2}

In this section, we develop a static equilibrium model of a productive economy composed of \( N \) sectors, where each sector is represented by a single firm that produces a distinct good. The economy's price system is determined by imposing market-clearing conditions: the total supply of each good must equal the total demand, which includes both final consumption and intermediate input requirements from all sectors. We then formalize the interaction among firms through the notion of \emph{competitive equilibrium}, defined as a collection of prices and input allocations such that (i) each firm maximizes its profit given prices, and (ii) all markets clear.

We establish the existence of such a competitive equilibrium by characterizing it as the saddle point of a central planner’s minimax problem. This variational formulation provides an economic interpretation aligned with the First Fundamental Theorem of Welfare Economics, emphasizing the efficiency of equilibrium allocations.

Each sector faces external demand from consumers and may also compete with an external source of supply. To model this, we introduce sector-specific net demand functions $D_i: \R_{\geq 0} \to \R$, where $D_i(P)$ denotes the difference between external demand and external supply for the $i$-th sector at a given price $P \geq 0$. These functions may take both positive and negative values and are assumed to satisfy the following conditions.

\begin{assumption}[Properties of net demand functions]\label{ass:demand}
For every $i = 1, \dots, N$, the net demand function $D_i$ is strictly decreasing and continuous on $\R_{>0}$. Moreover, it satisfies the following asymptotic conditions:
\[
\lim_{P \to 0} D_i(P) = \infty, \quad \lim_{P \to \infty} D_i(P) \leq 0.
\]
\end{assumption}

This assumption encompasses both a closed economy without exogenous supply and an open economy with external supply. In the closed economy case, we have $D_i(P) > 0$ for all $P \geq 0$, and the limiting value as the price tends to infinity approaches zero.

\begin{example}[Power demand functions]\label{ex:power_demad}
The following examples illustrate modifications of the standard power demand function that account for the presence of external supply, while satisfying Assumption~\ref{ass:demand}.

A first class of net demand functions is given by:
\[
D_i(P) = a_i P^{-\epsilon_i} - b_i, \quad \text{with } a_i, \epsilon_i > 0, \quad b_i \geq 0,
\]
where the first term models consumer demand with elasticity $\epsilon_i$, and the constant term $b_i$ represents a finite external supply.

A second variant incorporates a price-dependent external supply:
\[
D_i(P) = a_i P^{-\epsilon_i} - b_i P^{\delta_i}, \quad \text{with } a_i, \epsilon_i, \delta_i > 0, \quad b_i \geq 0.
\]
Here, the external supply grows unboundedly with the price, so that the net demand diverges to $-\infty$ as $P \to \infty$.

In both cases, the closed economy scenario with no external supply is recovered by setting $b_i = 0$.
\end{example}

Let $\cQ = (q_{ij})_{i,j=1,\dots,N}$ represent the matrix of quantities exchanged among companies. Using this notation, we define an allocation matrix $(\cQ, \mathbf{E}, \mathbf{L}) \in \R^{N(N+2)}$. 

Under the assumption of no market frictions, the price vector must satisfy the following market clearing condition: 
\begin{equation}\label{eq:clearing}
    F_i(\mathbf{q}_i,E_i,L_i)=D_i(P_i)+\sum^N_{j=1} q_{ji},\quad i=1,\dots,N.
\end{equation}
for any fixed $(\cQ,\mathbf{E},\mathbf{L})\in\R^{N(N+2)}$.

The left-hand side of the equation represents the total endogenous supply of good $i$, which equals the output of  sector $i$. The right-hand side represents the total net demand for good $i$, comprising two components: the exogenous net demand, driven by the demand function $D_i$, and the endogenous or intermediate demand, which reflects the quantity of good $i$ demanded by all sectors, including sector $i$ itself.

We now introduce the notion of competitive equilibrium within our framework.

\begin{definition}[Competitive equilibrium]\label{def:Nash_equilibrium}
A competitive equilibrium (CE) is a vector $(\mathbf{P}^*,\cQ^*,\mathbf{E}^*,\mathbf{L}^*)\in\R^N_{\geq 0}\times\R^{N(N+2)}_{\geq 0}$ of prices and inputs such that, for each $i=1,\dots,N$:
\begin{itemize}
    \item[(i)] The price $P^{*,i}$ satisfies the market clearing condition \eqref{eq:clearing}.
    \item[(ii)] The input vector $(\mathbf{q}^*_i,E^*_i,L^*_i)$ solves the profit maximization problem \eqref{prob:profit_max}.   
\end{itemize}
\end{definition}

To establish the existence of a CE, we consider the following minimax problem:
\begin{equation}\label{prob:minimax}
\min_{\mathbf{P}}\sum^N_{i=1}\max_{\mathbf{q}_i,E_i,L_i}\Bigg\{P_i F_i(\mathbf{q}_i,E_i,L_i)-\sum^N_{j=1}P_jq_{ij}-P_E E_i-W_i(L_i) -\Delta_i(P_i)\Bigg\},    
\end{equation}
where $\Delta_i(P)\coloneqq\int^P_1D_i(z)\,dz$ denotes the consumer surplus function of sector $i$ up to a given price $P \geq 0$.

By interchanging the summation and maximization operators, we reformulate Problem~\eqref{prob:minimax} as:
\begin{equation}\label{prob:lagrangian}
\min_{\mathbf{P}}\max_{\cQ,\mathbf{E},\mathbf{L}}\cL(\mathbf{P},\cQ,\mathbf{E},\mathbf{L}),   
\end{equation}
where $\cL:\R^N_{\geq 0}\times\R^{N(N+2)}\to\R$ is the Lagrangian function defined by:
\begin{equation}\label{eq:lagrangian}
\cL(\mathbf{P},\cQ,\mathbf{E},\mathbf{L})\coloneqq\sum^N_{i=1}\Bigg\{P_i F_i(\mathbf{q}_i,E_i,L_i)-\sum^N_{j=1}P_jq_{ij}-P_E E_i-W_i(L_i) -\Delta_i(P_i)\Bigg\}. 
\end{equation}
The Lagrangian function captures the total surplus of the economy: it aggregates firms' profits while subtracting the consumer surplus terms $\Delta_i(P_i)$, which appear with a negative sign to reflect that consumers only demand goods in this model and do not share in firms' profits.

\begin{remark}[On the consumer surplus function]\label{rem:CDF}
Under Assumption \ref{ass:demand}, the consumer surplus functions $\Delta_i:\R_{\geq 0}\to\R$ are well-defined for all non-negative prices, with the convention: 
$$
\Delta_i(P)=-\int^1_P D_i(z) \, dz, \quad\text{for all $P\leq 1$}. 
$$
In particular, $\lim_{P\to 0^+}\Delta_i(P)$ exists, though it may equal $-\infty$. Consequently, $\Delta_i$ is upper semi-continuous over $\R_{\geq 0}$.  

Additionally, since the demand functions $D_i$ are continuous on $\R_{>0}$, each consumer surplus function $\Delta_i$ is continuously differentiable on this domain.

Finally, because each demand function $D_i$ is strictly decreasing, the consumer surplus functions $\Delta_i$ are strictly concave.
\end{remark}

From standard convex analysis (see, e.g., \cite[Ch.~6]{Ekeland99}), the following inequality holds:
\begin{equation}\label{eq.minimax_inequality}
\sup_{\cQ,\mathbf{E},\mathbf{L}}\inf_{\mathbf{P}}\cL(\mathbf{P},\cQ,\mathbf{E},\mathbf{L})\leq\inf_{\mathbf{P}}\sup_{\cQ,\mathbf{E},\mathbf{L}}\cL(\mathbf{P},\cQ,\mathbf{E},\mathbf{L}),
\end{equation}
with strict inequality generally occurring. The right-hand side corresponds to the competitive equilibrium formulation, where each agent optimizes its own objective given prevailing market prices and the decisions of other agents. In contrast, the left-hand side — corresponding to the minimax problem \eqref{prob:lagrangian} — reflects the perspective of a central planner who simultaneously allocates inputs and sets prices in order to maximize total economic surplus, as represented by the Lagrangian \eqref{eq:lagrangian}. 
 
We now recall the following well-known definition, tailored to our context.
\begin{definition}[Saddle point]\label{def:saddle_point}
A vector $(\mathbf{P}^*,\cQ^*,\mathbf{E}^*,\mathbf{L}^*)\in\R^N_{\geq 0}\times\R^{N(N+2)}_{\geq 0}$ is called a saddle point of $\cL$ on $\R^N_{\geq 0}\times\R^{N(N+2)}_{\geq 0}$ if:
\[
\cL(\mathbf{P}^*,\cQ,\mathbf{E},\mathbf{L})\leq \cL(\mathbf{P}^*,\cQ^*,\mathbf{E}^*,\mathbf{L}^*)\leq\cL(\mathbf{P},\cQ^*,\mathbf{E}^*,\mathbf{L}^*)
\]
for all $\mathbf{P}\in\R^N_{\geq 0}$ and $(\cQ,\mathbf{E},\mathbf{L})\in\R^{N(N+2)}_{\geq 0}$. 

In particular, a vector $(\mathbf{P}^*,\cQ^*,\mathbf{E}^*,\mathbf{L}^*)$ is a saddle point if and only if satisfies the Minimax equality:
\begin{equation}\label{eq:minimax}
\cL(\mathbf{P}^*,\cQ^*,\mathbf{E}^*,\mathbf{L}^*)=\max_{\cQ,\mathbf{E},\mathbf{L}}\min_{\mathbf{P}}\cL(\mathbf{P},\cQ,\mathbf{E},\mathbf{L})=\min_{\mathbf{P}}\max_{\cQ,\mathbf{E},\mathbf{L}}\cL(\mathbf{P},\cQ,\mathbf{E},\mathbf{L}).
\end{equation}
\end{definition} 

The following lemma establishes a link between the notion of a saddle point of $\mathcal{L}$ and competitive equilibrium. The proof is provided in Appendix~\ref{app:1}. 

\begin{lemma}\label{lem:saddle_point}
Suppose that Assumptions \ref{ass:demand} holds. Then, $(\mathbf{P}^*,\cQ^*,\mathbf{E}^*,\mathbf{L}^*)\in\R^N_{> 0}\times \R^{N(N+2)}_{\geq 0}$ is a saddle point of $\cL$ if and only if it is a competitive equilibrium.  
\end{lemma}

To establish the existence of a CE, an additional assumption is required to ensure market viability. Specifically, there must exist a vector of positive prices and a production schedule that allows for profitable production in all sectors when labor and emission costs are zero.  

\begin{assumption}[Market Viability]\label{ass:participation}  
There exists a triple \( (\overline{\cQ},\overline{\mathbf{E}},\overline{\mathbf{L}}) \in \mathbb{R}^{N(N+2)}_{\geq 0} \) and a vector of positive prices \( \overline{\mathbf{P}} \in \mathbb{R}^N_{> 0} \) such that, for each \( i=1,\dots,N \),  
\[
\overline{P}_i F_i(\overline{\mathbf{q}}_i, \overline{E}_i, \overline{L}_i) - \sum_{j=1}^N \overline{P}_j \overline{q}_{ij} > 0.
\]
\end{assumption}  

Alternatively, one may assume that external supply is unbounded, which guarantees that net demand decreases without bound as the price grows arbitrarily large.  

\begin{assumption}[Unbounded External Supply]\label{ass:supply}  
For each \( i=1,\dots,N \), the net demand function \( D_i \) satisfies the following asymptotic condition:  
\[
\lim_{P\to\infty} D_i(P) = -\infty.
\]
\end{assumption}

Each of the above assumptions provides a sufficient condition for the existence of a competitive equilibrium, as summarized in the following theorem. 

\begin{theorem}\label{th:Nash_equilibrium}  
Suppose that the production and labor cost functions satisfy either Assumptions \ref{ass:prod_function} (a) and \ref{ass:wage_function}, or Assumptions \ref{ass:prod_function} (b), \ref{ass:wage_function}, and \ref{ass:wage_function2}. Furthermore, assume that the net demand functions satisfy Assumption \ref{ass:demand} and that either Assumption \ref{ass:participation} or \ref{ass:supply} holds.

Then, there exists a CE \( (\mathbf{P}^*,\cQ^*,\mathbf{E}^*,\mathbf{L}^*) \in \mathbb{R}^N_{> 0} \times \mathbb{R}^{N(N+2)}_{\geq 0} \), and the corresponding price system is unique. Moreover, if the production functions \( F_i \) are strictly concave and the labor cost functions \( W_i \) are strictly convex, then the CE is also unique.  
\end{theorem}  

\begin{proof}
By Lemma \ref{lem:saddle_point}, establishing the existence and uniqueness of a CE is equivalent to proving the existence and uniqueness of a saddle point for $\cL$ on $\R^N_{>0}\times\R^{N(N+2)}_{\geq 0}$. To this end, we verify the assumptions of the minimax theorem (see Proposition VI.2.4 in \cite{Ekeland99}).

We first express the Lagrangian function $\cL$ as:
\begin{equation}\label{eq:lagrangian_sum}
\cL(\mathbf{P},\cQ,\mathbf{E},\mathbf{L})=\sum^N_{i=1}\ell_i(\mathbf{P},\mathbf{q}_i,E_i,L_i),
\end{equation}
where the functions $\ell_i:\R^N_{\geq 0}\times\R^{N+2}_{\geq 0}\to\R$ are defined by:
\[
\ell_i(\mathbf{P},\mathbf{q}_i,E_i,L_i)\coloneqq P_i F_i(\mathbf{q}_i,E_i,L_i)-\sum^N_{j=1}P_jq_{ij}-P_E E_i-W_i(L_i) -\Delta_i(P_i),\quad i=1,\dots,N.
\]

Given the semi-continuity assumptions on $F_i$ and $W_i$, the function $\ell_i$ is upper semi-continuous in $(\mathbf{q}_i, E_i, L_i)$. Additionally, since $\Delta_i$ is upper semi-continuous, $\ell_i$ is lower semi-continuous with respect to $\mathbf{P}$. Consequently, the Lagrangian $\cL$ satisfies the necessary continuity properties: it is upper semi-continuous with respect to the inputs and lower semi-continuous with respect to the prices.

Furthermore, since $F_i$ and $\Delta_i$ are concave, and $W_i$ is convex, each function $\ell_i$ is concave in $(\mathbf{q}_i,E_i,L_i)$ and convex in $\mathbf{P}$. Consequently, the Lagrangian function $\cL$ is concave with respect to $(\cQ,\mathbf{E},\mathbf{L})$ and convex with respect to $\mathbf{P}$. 

Next, as demonstrated in the proof of Proposition \ref{prop:profit_max}, the profit function $\Pi_i$ tends to $-\infty$ as $\norm{(\mathbf{q}_i,E_i,L_i)}\to\infty$, for any fixed price vector $\mathbf{P}\in \mathbb R_{> 0}^N$. Consequently, for each $i=1,\dots,N$,
\[
\ell_i(\mathbf{P},\mathbf{q}_i,E_i,L_i)\to-\infty,\quad\text{as}\quad\norm{(\mathbf{q}_i,E_i,L_i)}\to\infty.
\]
This implies that 
\[
\cL(\mathbf{P},\cQ,\mathbf{E},\mathbf{L})\to-\infty,\quad\text{as}\quad\norm{(\cQ,E,L)}\to\infty.
\]
for any price vector $\mathbf{P}\in\R^N_{\geq 0}$.

To complete the proof of the existence of a saddle point, it remains to establish that the Lagrangian satisfies the following coercivity condition:  
\begin{equation}\label{eq:L_coercive}
\mathcal{L}(\mathbf{P}, \overline{\mathcal{Q}}, \overline{\mathbf{E}},\overline{\mathbf{L}}) \to +\infty \quad \text{as } \|\mathbf{P}\| \to \infty
\end{equation}
for some $(\overline{\mathcal{Q}},\overline{\mathbf{E}},\overline{\mathbf{L}}) \in \mathbb{R}^{N(N+2)}$.  

Assume that Assumption \ref{ass:participation} holds. Let $(\overline{\mathcal{Q}},\overline{\mathbf{E}},\overline{\mathbf{L}})$ and $\overline{\mathbf{P}}$ be as given in that assumption. Define a sector \( k \) such that  
\[
\frac{P_k}{\overline{P}_k} \geq \frac{P_i}{\overline{P}_i}, \quad \text{for all } i=1,\dots,N.
\]
Using this selection, we derive the following lower bound for the Lagrangian:  
\begin{align*}
\mathcal{L}(\mathbf{P}, \overline{\mathcal{Q}}, \overline{\mathbf{E}},\overline{\mathbf{L}}) &\geq \frac{P_k}{\overline{P}_k} \Bigg( \overline{P}_k F_k(\overline{\mathbf{q}}_k,\overline{E}_k,\overline{L}_k)-\sum^N_{j=1} \overline{P}_j \overline{q}_{kj} \Bigg) - C - \sum_{i=1}^N \Delta_i(P_i),
\end{align*}
where \( C < \infty \) is a constant independent of \( \mathbf{P} \).  

Furthermore, by Assumption \ref{ass:demand}, for every \( \varepsilon>0 \), there exists a constant \( C' \) such that  
\[
\Delta_i(P) \leq C' + \varepsilon P, \quad \text{for all } P > 0.
\]  
Since \( P_k \geq c\|\mathbf{P}\| \) for some constant \( c > 0 \), it follows that  
\[
\mathcal{L}(\mathbf{P}, \overline{\mathcal{Q}}, \overline{\mathbf{E}},\overline{\mathbf{L}}) \geq  c\|\mathbf{P}\| \Bigg( \overline{P}_k F_k(\overline{\mathbf{q}}_k,\overline{E}_k,\overline{L}_k) - \sum^N_{j=1} \overline{P}_j \overline{q}_{kj} - \frac{\varepsilon}{c} \Bigg) - C - C'.
\]
By Assumption~\ref{ass:participation}, \( \varepsilon \) can be chosen sufficiently small to ensure that the term in parentheses remains strictly positive. Consequently, the coercivity condition \eqref{eq:L_coercive} is satisfied.

Assume that the net demand functions satisfy the asymptotic condition specified in Assumption~\ref{ass:supply}. Then, for every \( \varepsilon > 0 \), there exist constants \( C \in \mathbb{R} \) and \( \hat{P} > 0 \) such that
\[
\Delta_i(P) \leq C - \varepsilon P, \quad \text{for all } P \geq \hat{P}.
\]

This implies that, for sufficiently large \( \|\mathbf{P}\| \), the Lagrangian satisfies the following lower bound:  
\[
\mathcal{L}(\mathbf{P},\mathcal{Q}, \mathbf{E},\mathbf{L}) \geq \sum_{i=1}^{N} \Bigg\{ P_i F_i(\mathbf{q}_i, E_i, L_i) - \sum_{j=1}^{N} q_{ij} P_j+ \varepsilon P_i \Bigg\} + C'
\]
for some constant \( C' \) that is independent of \( \mathbf{P} \).  

By choosing \( \varepsilon \) sufficiently large, the coercivity condition \eqref{eq:L_coercive} is satisfied.

This implies that the assumptions of Proposition VI.2.4 in \cite{Ekeland99} hold, allowing us to conclude that  
\begin{align*}
\max_{\cQ,\mathbf{E},\mathbf{L}}\min_{\mathbf{P}}\cL(\mathbf{P},\cQ,\mathbf{E},\mathbf{L})=\min_{\mathbf{P}}\max_{\cQ,\mathbf{E},\mathbf{L}}\cL(\mathbf{P},\cQ,\mathbf{E},\mathbf{L}),\label{infmax}
\end{align*}
and that a saddle point $(\mathbf{P}^*,\cQ^*,\mathbf{E}^*,\mathbf{L}^*)\in\R^N_{\geq  0}\times \R^{N(N+2)}_{\geq 0}$ exists.  

To apply Lemma~\ref{lem:saddle_point}, it remains to verify that \( \mathbf{P}^* \in \mathbb{R}^N_{> 0} \). Suppose, for contradiction, that \( P^*_i = 0 \) for some \( i \in \{1,\dots,N\} \). For an arbitrary \( h > 0 \), define the perturbed price vector \( \mathbf{P}^*(h) \) by
\[
P^*_j(h) = P^*_j \quad \text{for } j \neq i, \quad \text{and} \quad P^*_i(h) = h.
\] 

Since the Lagrangian \( \mathcal{L} \) is continuously differentiable on \( \mathbb{R}^N_{> 0} \), the mean value theorem implies the existence of a vector \( \mathbf{v}(h) \in \mathbb{R}^N_{\geq 0} \) such that  
\[
v_j(h) = P^*_j \quad \text{for } j \neq i, \quad \text{and} \quad v_i(h) \in (0, h),
\]
for which the following first-order expansion holds:
\begin{align*}
\cL(\mathbf{P}^*(h), \cQ, \mathbf{E}, \mathbf{L}) 
&= \cL(\mathbf{P}^*, \cQ, \mathbf{E}, \mathbf{L}) + h \frac{\partial}{\partial P_i} \cL(\mathbf{v}(h), \cQ, \mathbf{E}, \mathbf{L}) \\
&= \cL(\mathbf{P}^*, \cQ, \mathbf{E}, \mathbf{L}) + h \left( F_i(\mathbf{q}_i, E_i, L_i) - \sum_{j=1}^N q_{ji} - D_i(v_i(h)) \right).
\end{align*}

By Assumption~\ref{ass:demand}, \( h \) can be chosen sufficiently small so that the expression in parentheses on the second line is strictly negative. Consequently,  
\[
\mathcal{L}(\mathbf{P}^*(h), \mathcal{Q}, \mathbf{E}, \mathbf{L}) < \mathcal{L}(\mathbf{P}^*, \mathcal{Q}, \mathbf{E}, \mathbf{L}),
\]
contradicting the optimality of \( \mathbf{P}^* \) as a minimizer of \( \mathcal{L} \). Therefore, we conclude that \( \mathbf{P}^* \in \mathbb{R}^N_{> 0} \).

Finally, since the Lagrangian is strictly convex with respect to $\mathbf{P}$, the optimal price system $\mathbf{P}^*$ is unique. Furthermore, if the production functions $F_i$ are strictly concave and the labor cost functions $W_i$ are strictly convex, then the Lagrangian function $\cL$ is strictly concave with respect to $(\cQ,\mathbf{E},\mathbf{L})$. These strict conditions ensure that if a saddle point — and hence a CE — exists, it must be unique.    
\end{proof}

\begin{remark}[Connection to welfare economics]\label{rem:welfare}
Theorem~\ref{th:Nash_equilibrium} establishes an equivalence between a competitive economy and a central planner, as reflected in the two sides of the minimax identity \eqref{eq:minimax}. That is, whenever a saddle point exists, the decentralized market equilibrium coincides with the solution chosen by a social planner seeking to maximize total welfare. In particular, the price system that emerges from the decentralized mechanism — where supply and demand are balanced — matches the price vector selected by the planner in solving the centralized problem \eqref{prob:lagrangian}.

This one-to-one correspondence between competitive equilibria and saddle points of \( \mathcal{L} \) on \( \mathbb{R}^N_{> 0} \times \mathbb{R}^{N(N+2)}_{\geq 0} \) embodies the spirit of the first fundamental theorem of welfare economics, reflects the spirit of the first fundamental theorem of welfare economics, in which competition and the absence of externalities are key to ensuring that decentralized market allocations are Pareto efficient (see \cite[Proposition 10.D.1]{Mas-colell95}).

In our setting, in particular, externalities are fully internalized through the carbon pricing mechanism.
\end{remark}

\section{Mean-field game approach to a dynamic market model with defaults}\label{sec:3}

In this section, we extend the framework developed in Sections \ref{sec:1} and \ref{sec:2} to a dynamic market model over the time interval $[0,T]$. The economy consists of $N$ sectors, each comprising a continuum of firms that may default at uncertain future times. This formulation, which models an economy with multiple sectors and infinitely many firms per sector, has been explored in the context of portfolio optimization, e.g., in \cite{borkar2010mckean}. As in that work and many others applying mean-field game (MFG) theory, we do not analyze the convergence from a finite-agent system to the mean-field limit, but instead work directly with the limiting model.

We consider an endogenous default framework in which each individual agent in each sector optimally selects its production inputs and the timing of market exit to maximize expected profits over $[0,T]$.

To streamline the equilibrium analysis, we adopt the linear programming approach introduced in \cite{bouveret2020mean} and \cite{Dumitrescu2021LP}, which reformulates the firms’ optimal stopping problem as an equivalent linear program. The central objects of interest are the distributions of states (labor costs) for non-defaulted firms in each sector, as well as the joint distributions of exit times and labor cost trajectories —referred to as exit measures.

A linear programming MFG Nash equilibrium is defined as a configuration of prices, state distribution flows, and exit measures that jointly satisfy a mean-field version of the market-clearing condition \eqref{eq:clearing} and the generic agent’s profit-maximization problem in each sector. The section concludes by proving the existence of such an equilibrium and the uniqueness of the associated price system.

\subsection{The optimal default problem of the generic firm}\label{sec:3.1}

Throughout this section, we fix a sector \( i \in \{1, \dots, N\} \) and focus on a generic firm within that sector.

Let $(\Omega,\cF,\bbF,\bbP)$ be a stochastic basis, where the filtration $\bbF = (\cF_t)_{t \in [0,T]}$ satisfies the usual conditions of right-continuity and completeness.

Let $\mathbf{q}^i_t \coloneqq (q^{i,j}_t)_{j=1,\dots,N}$ denote the vector of intermediate inputs purchased, $E^i_t$ the emissions produced, and $L^i_t$ the labor employed by the generic firm in sector $i$ per unit of time at time $t \in [0,T]$.

The total output of this firm per unit of time at time $t$ is given by
\begin{equation}\label{eq:output_firm}
    F_i(\mathbf{q}^i_t, E^i_t, L^i_t),\quad t \in [0,T].
\end{equation}
We assume that the production function $F_i$ satisfies the conditions of Assumption  \ref{ass:prod_function} (b)\footnote{In the static model, the main results hold under both Assumption \ref{ass:prod_function} (a) and Assumption \ref{ass:prod_function} (b). However, in the dynamic stochastic framework with mean-field interaction, the equilibrium analysis becomes significantly more challenging under Assumption \ref{ass:prod_function} (a). Assuming production functions with constant returns to scale allows the maximized profit of each firm to be expressed explicitly in terms of its labor cost function, as illustrated in Example \ref{ces.ex}. This simplification streamlines both the theoretical analysis of equilibria and their numerical approximation.}. 

We introduce an $\bbF$-adapted stochastic process $X^i = (X^i_t)_{t \in [0,T]}$ taking values in a sector-specific state space $\cO_i \subseteq \R_{> 0}$, and satisfying the following stochastic differential equation:
\begin{equation}\label{eq:MF_SDE}
    dX^i_t = \alpha_i(t, X^i_t)\,dt + \sigma_i(t, X^i_t)\,dW^i_t, \quad X^i_0 = x^i \in \cO_i, \quad t \in [0,T],
\end{equation}
where $(W^i)_{i=1,\dots,N}$ are independent standard Brownian motions. The functions $\alpha_i$ and $\sigma_i$, representing the drift and volatility of $X^i$, satisfy the technical conditions stated below.

\begin{assumption}[Properties of Drift and Volatility Functions]\label{ass:SDE}
The functions $\alpha_i\colon [0,T] \times \cO_i \to \R$ and $\sigma_i\colon [0,T] \times \cO_i \to \R_{> 0}$ ensure the existence of a unique strong solution to Equation \eqref{eq:MF_SDE}, such that:
\[
\sup_{t \in [0,T]} \bbE\big[\abs{X^i_t}^p\big] < \infty, \quad \bbP(\tau_{\cO_i} > T) = 1,
\]
for some $p \geq 1$, where $\tau_{\cO_i}$ denotes the first exit time of $X^i$ from the domain $\cO_i$.

In addition, the following conditions hold:
\begin{enumerate}
    \item The functions $\alpha_i(t,x)$ and $\sigma_i(t,x)$ are jointly measurable and continuous in $x$, uniformly in $t \in [0,T]$;
    \item There exist constants $M > 0$ and $\beta \in [0,1]$ such that, for all $(t,x) \in [0,T] \times \cO_i$,
    \[
    \abs{\alpha_i(t,x)} \leq M, \quad \sigma_i(t,x)^2 \leq M(1 + x^{\beta}).
    \]
\end{enumerate}
\end{assumption}

\begin{remark}\label{rem:CIR}
Assumption~\ref{ass:SDE} does not require Lipschitz continuity of the drift and volatility functions. While Lipschitz continuity — combined with measurability and sublinear growth — is a standard sufficient condition for the existence of a unique strong solution to the SDE~\eqref{eq:MF_SDE}, we relax this requirement to allow for a broader class of dynamics. In particular, this class includes processes such as the Cox-Ingersoll-Ross (CIR) process:
\[
dX^i_t = \alpha_i(\theta_i-X^i_t)\,dt + \sigma_i \sqrt{X^i_t}\,dW^i_t, \quad X^i_0 = x^i \in \R_{>0}, \quad t \in [0,T],
\]
where \( \alpha_i \) denotes the mean-reversion speed and \( \theta_i \) the long-run mean level.

For the CIR process, it is well known that pathwise uniqueness and non-negativity of solutions are guaranteed under Feller's condition:
\[
2\alpha_i \theta_i > \sigma_i^2.
\]
A detailed discussion can be found in \cite[Ch.~6]{JeanblancYorChesney2009}.
\end{remark}

We assume that prices and the carbon tax are deterministic and represented as time-dependent functions:
\[
P_i\colon[0,T]\to\R_{\geq 0}, \quad P_E\colon[0,T]\to\R_{>0}.
\]
Firms are assumed to operate under full information: at each time $t$, they observe the history of the price vector, carbon tax, and their own state process up to time $t$. 

We define the instantaneous cost function of the generic firm in sector~$i$ by
\begin{equation}\label{eq:cost_sector}
C_i(\mathbf{P}, \mathbf{q}_i, E_i, L_i, P_E, x) = \sum_{j=1}^N P_j q_{ij} + P_E E_i + W_i(L_i) + xL_i,
\end{equation}
for any price vector $\mathbf{P}=(P_1,\dots,P_N) \in \R^N_{\geq 0}$, input vector $(\mathbf{q}_i, E_i, L_i) \in \R^{N+2}_{\geq 0}$, carbon price $P_E \in \R_{>0}$, and state variable $x \in \cO_i$.

The function $W_i \colon \R_{\geq 0} \to \R_{>0}$, representing the baseline labor cost in sector $i$, satisfies Assumptions~\ref{ass:wage_function} and~\ref{ass:wage_function2}. The final term in Equation~\eqref{eq:cost_sector} captures additional fluctuations in wage levels driven by the state process $X^i$, which is why we also refer to it as the \emph{labor cost process}.

The instantaneous profit function of sector~$i$ is defined by
\begin{equation}\label{eq:gain_sector}
\Pi_i(\mathbf{P}, \mathbf{q}_i, E_i, L_i, P_E, x) = P_i F_i(\mathbf{q}_i, E_i, L_i) - C_i(\mathbf{P}, \mathbf{q}_i, E_i, L_i, P_E, x),
\end{equation}
for any price vector $\mathbf{P} \in \R^N_{\geq 0}$, input vector $(\mathbf{q}_i, E_i, L_i) \in \R^{N+2}_{\geq 0}$, carbon price $P_E \in \R_{> 0}$, and state variable $x \in \cO_i$. The corresponding maximized instantaneous profit function is given by
\begin{equation}\label{eq:max_inst_profit}
\overline{\Pi}_i(\mathbf{P}, P_E, x) \coloneqq \max_{\mathbf{q}_i, E_i, L_i} \Pi_i(\mathbf{P}, \mathbf{q}_i, E_i, L_i, P_E, x).
\end{equation}

The following proposition shows that the function $\overline{\Pi}_i$ is well defined and admits an explicit expression in terms of $W_i^*$. The proof is identical to that of Propositions~\ref{prop:profit_max} and~\ref{prop:CRS}.

\begin{proposition}\label{prop:max_inst_profit}  
Suppose that the production and labor cost functions satisfy Assumptions~\ref{ass:prod_function}(b), ~\ref{ass:wage_function}, and~\ref{ass:wage_function2}. Then the following statements hold:
\begin{enumerate}
    \item[(i)] The optimization problem \eqref{eq:max_inst_profit} admits a solution \(\mathbf{x}^*_i = (\mathbf{q}^*_i, E^*_i, L^*_i) \in \R^{N+2}_{\geq 0}\) for all price vectors $\mathbf{P} \in \R^N_{> 0}$, carbon prices $P_E \in \R_{> 0}$, and states $x \in \cO_i$. Moreover, if the production function $F_i$ is strictly concave and the labor cost function $W_i$ is strictly convex, the solution is unique.
    
    \item[(ii)] The maximized instantaneous profit function \(\overline{\Pi}_i\) is given by
    \[
    \overline{\Pi}_i(\mathbf{P}, P_E, x) = W_i^*\left(\widetilde{\Pi}_i(\mathbf{P}, P_E) - x\right),
    \]
    where \(\widetilde{\Pi}_i\) is defined in Equation~\eqref{pitilde}.
    
    \item[(iii)] If the labor cost function $W_i$ is strictly convex and differentiable, the optimal inputs are given by
    \[
    L^*_i = (W_i')^{-1}\left(\widetilde{\Pi}_i(\mathbf{P}, P_E) - x\right), \quad q^*_{ij} = \tilde{q}^*_{ij} L^*_i, \, j=1,\dots,N, \quad E^*_i = \tilde{E}^*_i L^*_i,
    \]
    where $(\tilde{q}^*_{ij})_{j=1,\dots,N}$ and $\tilde{E}^*_i$ are the solutions of the auxiliary problem defined by Equation~\eqref{pitilde}.
\end{enumerate}
\end{proposition}

As outlined in the introduction, each firm may default at a random time. The default mechanism is modeled as an optimal stopping problem, in which firms simultaneously choose the trajectory of their input mix and the optimal timing of market exit to maximize expected profits up to default. We assume that the generic firm in sector \( i \) issues debt contracts with maturity \( T \), which provide bondholders with a constant coupon stream \( \kappa_i \) over the interval \( [0, T] \). Upon default, the firm incurs a fixed cost \( K_i \), subject to exponential depreciation at rate \( \gamma_i \).

The expected profit maximization problem for the generic firm in sector \( i \in \{1,\dots,N\} \) is given by:
\begin{equation*}
\max_{\tau_i \in \cT([0,T])} \max_{(\mathbf{q}^i_t,E^i_t,L^i_t)_{t \geq 0}} 
\bbE_x \left[ \int_0^{\tau_i} e^{-r t} \left\{ \Pi_i(\mathbf{P}(t), P_E(t), X^i_t) - \kappa_i \right\} \, dt 
- K_i e^{-(r+\gamma_i)\tau_i} \right],
\end{equation*}
where \( \cT([0,T]) \) denotes the set of \( \bbF \)-stopping times on \( [0,T] \), \( \bbE_x \) denotes expectation conditional on \( X^i_0 = x \), and \( r \) is the capital discount rate.

Using the definition of the maximized instantaneous profit \eqref{eq:max_inst_profit}, this problem can be equivalently rewritten as:
\begin{equation}\label{prob:firm}
v_i(\mathbf{P}, P_E, x) \coloneqq \max_{\tau_i \in \cT([0,T])} 
\bbE_x \left[ \int_0^{\tau_i} e^{-r t} \left\{ \overline{\Pi}_i(\mathbf{P}(t), P_E(t), X^i_t) - \kappa_i \right\} dt 
- K_i e^{-(r+\gamma_i)\tau_i} \right].
\end{equation}

The value function \( v_i(\mathbf{P}, P_E, x) \) represents the equity value of firm \( i \), net of the expected discounted liquidation cost incurred at the time of exit. In problems with endogenous default, the optimal default policy typically takes the form of a threshold strategy (see for example \cite{leland1996optimal},\cite{belanger2004general}, \cite{schmidt2008structural}, \cite{frey2009pricing} among many others). That is, the firm exits the industry and liquidates its assets immediately once the labor shock process \( (X^i_t)_{t \geq 0} \) crosses an endogenously determined upper barrier.

However, in optimal stopping mean-field games, it has been observed \citep{bertucci2018mixed} that an equilibrium with pure (threshold-type) stopping times does not always exist. To capture all equilibria, it is necessary to consider a broader class of strategies, namely randomized stopping times. This implies that, at equilibrium, agents may exit at random times within the interval between the minimal optimal stopping time \( \tau_- \) and the maximal optimal stopping time \( \tau_+ \), rather than exiting precisely at the first hitting time of a threshold. As a result, a characterization of equilibrium solely in terms of stopping times or default thresholds is not always feasible, and alternative methods have been developed.

In particular, the linear programming approach introduced by \cite{bouveret2020mean} circumvents the need to solve individual optimization problems by focusing directly on the evolution of population distributions. This avoids the technical challenges posed by randomized stopping strategies.

\subsection{Linear programming formulation}\label{sec:3.3}
 
Within the linear programming framework, we reformulate the optimal stopping problem~\eqref{prob:firm} as a linear program over a suitable space of measures. We begin by recalling the following definition from \cite{bouveret2020mean}, suitably adapted to our setting.

\begin{definition}[Space of Bounded Measure Flows]
Let $p \geq 1$ and $\cO \subseteq \R$. We define $\cV_p(\overline{\cO})$ as the space of flows $m = (m_t(\cdot))_{t \in [0,T]}$ of bounded measures on $\overline{\cO}$, satisfying the following properties:
\begin{enumerate}
\item For every $t \in [0,T]$, $m_t$ is a bounded measure on $\overline{\cO}$;
\item For every Borel set $A \in \cB(\overline{\cO})$, the mapping $t \mapsto m_t(A)$ is measurable;
\item The $p$-moment condition holds:
\[
\int_0^T \int_{\overline{\cO}} (1 + x^p) \, m_t(dx) dt < \infty.
\]
\end{enumerate}
Each such flow $m \colon[0,T] \times \cB(\overline{\cO}) \to \R_{\geq 0}$ induces a bounded measure on $[0,T] \times \overline{\cO}$ via the product $m_t(dx) dt$. We denote by $\cM_p([0,T]\times\overline{\cO})$ the space of bounded measures on $[0,T] \times \overline{\cO}$ that satisfy Condition (3). 

We equip $\cM_p([0,T]\times\overline{\cO})$ with the topology of weak convergence induced by continuous functions with $p$-growth in $x \in \overline{\cO}$, denoted by $\tau_p$. The space $\cV_p(\overline{\cO})$ is endowed with the corresponding topology induced by the weak convergence of associated measures, also denoted by $\tau_p$.
\end{definition}

\begin{remark}[On the topology $\tau_p$]\label{haussdorf}
The topology $\tau_p$ generalizes the classical weak topology by allowing test functions with polynomial growth. Specifically, a sequence $(\nu^n)_{n \geq 1} \subset \cM_p([0,T]\times\overline{\cO})$ converges to a limit $\nu$ in the topology $\tau_p$ if
\[
\int_0^T \int_{\overline{\cO}} f(t,x) \, \nu^n(dt,dx) \longrightarrow \int_0^T \int_{\overline{\cO}} f(t,x) \, \nu(dt,dx), \quad \text{as } n \to \infty,
\]
for all test functions $f\colon [0,T] \times \overline{\cO} \to \R$ that are measurable in $t$, continuous in $x$, and satisfy a $p$-growth condition of the form
\[
|f(t,x)| \leq C(1 + x^p), \quad \text{for all } (t,x) \in [0,T] \times \overline{\cO},
\]
for some constant $C > 0$.

The space $\cV_p(\overline{\cO})$, endowed with the topology $\tau_p$, is Hausdorff, locally convex, and metrizable. For further details, we refer the reader to \cite{dumitrescu2023linear}.
\end{remark}

For each sector \( i = 1, \dots, N \), we associate a pair \( (m^i, \mu^i) \in \cV_p(\overline{\cO}_i) \times \cM_p([0,T] \times \overline{\cO}_i) \), representing the dynamics of labor costs and defaults among firms in that sector. Specifically, \( m^i_t \) denotes the distribution of labor costs across firms that have not defaulted by time \( t \), while \( \mu^i \) encodes the joint distribution of default times and labor costs. Accordingly, we refer to the collection \( (m^i)_{i=1,\dots,N} \) as the \emph{occupation measure flows}, and \( (\mu^i)_{i=1,\dots,N} \) as the \emph{exit measures}.

The pair \( (m^i, \mu^i) \) satisfies a linear constraint induced by the dynamics of the labor cost process \( X^i \). Assuming that these measures admit sufficiently regular densities, the Fokker–Planck equation for the density flow \( (m^i_t)_{t\in [0,T]} \) reads:
\[
\frac{\partial m^i_t}{\partial t}(t,x) = \crL^*_i m^i_t(t,x) - \mu^i(t,x),
\]
where \( \crL^*_i \) is the adjoint of the generator of the diffusion process \( X^i \), given by
\begin{align*}
\crL_i f(t,x) &= \alpha_i(t,x) \frac{\partial f}{\partial x}(t,x) + \frac{1}{2} \sigma_i^2(t,x) \frac{\partial^2 f}{\partial x^2}(t,x), \\
\crL^*_i f(t,x) &= -\frac{\partial}{\partial x} \big(\alpha_i(t,x) f(t,x)\big) + \frac{1}{2} \frac{\partial^2}{\partial x^2} \big(\sigma_i^2(t,x) f(t,x)\big).
\end{align*}

In other words, the single-agent dynamics \eqref{eq:MF_SDE} is replaced by a Fokker–Planck equation governing the evolution of the population density, and the agent’s stopping time is replaced by the exit measure \( \mu^i \), which appears as a killing term and determines the exit behavior of the entire population.

In practice, however, the measures \( m^i \) and \( \mu^i \) may not possess the regularity required to satisfy the Fokker–Planck equation in its strong form. For instance, when agents exit upon hitting a default threshold, the measure \( \mu^i \) is supported on this threshold and lacks a density. Therefore, we work with a weak formulation of the equation, defined using appropriate test functions. The following definition is derived from the Fokker–Planck equation by multiplying it with a test function and applying integration by parts.

\begin{definition}\label{not:spaces}
For each $i = 1, \dots, N$ and initial distribution $m_0$, let $\cR^i(m_0)$ denote the subset of $\cV_p(\overline{\cO}_i)\times\cM_p([0,T]\times\overline{\cO}_i)$ consisting of all pairs $(m, \mu)$ satisfying the identity:
\begin{equation}\label{eq:weak_FP}
\begin{split}
\int_0^T \int_{\overline\cO_i} u(t,x)\, \mu(dt,dx) &=
\int_{\cO_i} u(0,x)\, m_0(dx) \\
&\quad + \int_0^T \int_{\overline\cO_i}
\left\{ \frac{\partial u}{\partial t}(t,x) +
\crL_i u(t,x) \right\}\,m_t(dx) dt,
\end{split}
\end{equation}
for all test functions $u \in \cC_b^{1,2}([0,T] \times \cO_i)$.
\end{definition}

A pair \( (m, \mu) \in \cR^i(m_0) \) characterizes, in a weak sense, the evolution of the distribution of the labor cost process for a generic firm in sector \( i \), governed by the stochastic differential equation \eqref{eq:MF_SDE} and stopped at a random time. Importantly, for any \( (m, \mu) \in \cR^i(m_0) \), the flow \( (m_t)_{t \in [0,T]} \) does not consist of probability measures: its total mass decreases over time as firms exit the system.

The following lemma establishes the compactness of the set $\cR^i(m_0)$, a key ingredient in proving the existence of an equilibrium; see Appendix \ref{app:2} for the proof.  
\begin{lemma}\label{lem:compactness} 
Suppose Assumption~\ref{ass:SDE} holds. Fix $i \in \{1,\dots,N\}$ and an initial distribution $m_0$ satisfying
\[
\int_{\cO_i}(1 + x^{p'})\,m_0(dx) < \infty \quad \text{for some } p' > p.
\]
Then the set $\cR^i(m_0)$ is weakly\footnote{Here, 'weakly' refers to compactness with respect to the topology $\tau_p \otimes \tau_p$.} compact in $\cV_p(\overline{\cO}_i) \times \cM_p([0,T]\times\overline{\cO}_i)$.
\end{lemma}
In the linear programming formulation of the optimal stopping problem, the optimization over stopping times is recast as an optimization over occupation flows and exit measures \( (m^i, \mu^i) \). Specifically, the firm's expected profit maximization problem becomes:
\begin{equation}\label{prob:MFG_profit}
\begin{split}
\max_{(m^i,\mu^i) \in \cR^i(m^i_0)} &\int_0^T \int_{\overline{\cO}_i} e^{-r t} \left\{ \overline\Pi_i(\mathbf{P}(t), P_E(t), x) - \kappa_i \right\} m^i_t(dx)\,dt\\
&-\int_0^T \int_{\overline{\cO}_i} K_i e^{-(r+\gamma_i)t} \mu^i(dt,dx).
\end{split}
\end{equation}

The connection between the original optimal stopping problem \eqref{prob:firm} and its linear programming counterpart \eqref{prob:MFG_profit} is rigorously established in \cite{bouveret2020mean} and \cite{Dumitrescu2021LP}. Specifically, defining the value functional
\begin{align*}
\Gamma_i(m^i, \mu^i,\mathbf{P},P_E) &\coloneqq \int_0^T \int_{\overline{\cO}_i} e^{-r t} \left\{ \overline\Pi_i(\mathbf{P}(t), P_E(t), x) - \kappa_i \right\} m^i_t(dx)dt\\
&\quad- \int_0^T \int_{\overline{\cO}_i} K_i e^{-(r+\gamma_i)t} \mu^i(dt,dx),
\end{align*}
it can be shown, under additional assumptions, that the following identity holds:
\[
\Gamma^*_i(\mathbf{P},P_E)\coloneqq\max_{m^i,\mu^i\in\cR^i(m^i_0)}\Gamma_i(m^i, \mu^i,\mathbf{P},P_E) = \int_{\cO_i} v_i(\mathbf{P},P_E,x) \, m^i_0(dx),
\]
where $v_i(\mathbf{P},P_E,x)$ is the value function previously introduced in the optimal stopping problem~\eqref{prob:firm}.

To complete the linear programming formulation, we specify the mechanism for equilibrium price formation. Analogous to the static model in Section~\ref{sec:2}, we introduce time-inhomogeneous net demand functions \( D_i\colon [0,T] \times \R_{\geq 0} \to \R \) for each good \( i = 1, \dots, N \). The value \( D_i(t, P_i(t)) \) at time \( t \) represents the net external demand for good \( i \) at the current time, defined as the difference between external demand and external supply. The time dependence accounts for evolving factors such as population growth or inflation.

We adopt the following assumption, which mirrors the demand conditions in the static model.

\begin{assumption}[Properties of Net Demand Functions]\label{ass:demand_uniform}
For each \( i = 1, \dots, N \), the net demand function \( D_i \) is continuous and strictly decreasing in the price variable. Moreover, it satisfies the limiting behavior:
\[
\lim_{P \to 0} D_i(t, P) = \infty, \qquad \lim_{P \to \infty} D_i(t, P) \leq 0,
\]
uniformly in \( t \in [0,T] \).
\end{assumption}

For each good \( i = 1, \dots, N \), the mean-field market-clearing condition reads:
\begin{equation}\label{eq:MFG_clearing}
\int_{\overline{\cO}_i} F_i(\mathbf{q}_i(t,x), E_i(t,x), L_i(t,x))\, m^i_t(dx) 
= D_i(t, P_i(t)) + \sum_{j=1}^N \int_{\overline{\cO}_i} q_{ji}(t,x)\, m^j_t(dx),
\end{equation}
for any input matrix function \( (\cQ, \mathbf{E}, \mathbf{L}) \colon [0,T] \times \cO_i \to \R_{\geq 0}^{N(N+2)} \), and for any vector of occupation measure flows \( \mathbf{m} \coloneqq (m^i)_{i=1,\dots,N} \).

The left-hand side of \eqref{eq:MFG_clearing} represents the total endogenous supply of good \( i \) from the mean-field firm at time \( t \), computed as the integral of the production function over the distribution of active (non-defaulted) agents. The right-hand side corresponds to the total net demand, including both external and endogenous contributions, again integrated against the current occupation measures.

\subsection{Linear programming MFG Nash equilibrium}
 
We now introduce the notion of market equilibrium within the linear programming MFG framework. To this end, let $\cR(\mathbf{m}_0)$ denote the set of all pairs of measure flows and bounded measures $(\mathbf{m},\boldsymbol{\mu})$ such that $(m^i,\mu^i)\in\cR^i(m^i_0)$ for each $i=1,\dots,N$.

\begin{definition}[Space of Price Vector Functions]\label{def:space_price} 
Let $q > 1$. Define $\cP_q\subset L^q([0,T];\R^N)$ as the space of measurable functions $\mathbf{P} \colon [0,T] \to \R^N_{\geq 0}$, endowed with the topology induced by the $L^q([0,T])$-norm. Let $\cP^+_q$ denote the subspace of $\cP_q$ consisting of functions whose components are strictly positive almost everywhere on $[0,T]$.
\end{definition}

\begin{definition}[Linear Programming MFG Nash Equilibrium]\label{def:MFG_equilibrium}
A linear programming MFG Nash equilibrium is a tuple $(\mathbf{P}^*, \cQ^*, \mathbf{E}^*, \mathbf{L}^*, \mathbf{m}^*, \boldsymbol{\mu}^*)$ consisting of:
\begin{itemize}
\item A price vector function $\mathbf{P}^* \in \cP^q$, for some $q>1$,
\item An input matrix function $(\mathcal Q^*, \mathbf E^*, \mathbf L^*) \colon [0,T] \times \cO_i \to \R^{N(N+2)}_{\geq 0}$,
\item A pair of measure flows and exit measures $(\mathbf{m}^*, \boldsymbol{\mu}^*) \in \cR(\mathbf{m}_0)$,
\end{itemize} 
such that, for each $i = 1, \dots, N$:
\begin{itemize}
\item[(i)] The price function $P^{*}_i$ satisfies the mean-field market-clearing condition \eqref{eq:MFG_clearing}. 
\item[(ii)] The input allocation $(\mathbf{q}^*_i(t,x), E^*_i(t,x), L^*_i(t,x))$ solves the firm's instantaneous profit maximization problem:
\[
\overline{\Pi}_i(\mathbf{P}(t), P_E(t), x) \coloneqq \max_{\mathbf{q}_i, E_i, L_i} \Pi_i(\mathbf{P}(t), \mathbf{q}_i, E_i, L_i, P_E(t), x),
\]
for each $(t,x) \in [0,T] \times \cO_i$.
\item[(iii)] The pair $(m^{i,*}, \mu^{i,*})\in\cR^i(m^i_0)$ solves the linear programming MFG profit maximization problem \eqref{prob:MFG_profit}.
\end{itemize}
\end{definition}

To establish the existence of a linear programming MFG Nash equilibrium, we adopt a strategy analogous to that used for the static model in Section \ref{sec:2}.

We define the Lagrangian function \( \mathcal{L} : \mathcal{P}_q \times \mathcal{R}(\mathbf{m}_0) \to \mathbb{R} \) as follows:
\begin{equation}\label{eq:MFG_lagrangian}
\begin{split}
\mathcal{L}(\mathbf{P}, \mathbf{m}, \boldsymbol{\mu}) &\coloneqq \sum_{i=1}^N \int_0^T \int_{\overline{\mathcal{O}}_i} e^{-r t} \left\{ \overline{\Pi}_i(\mathbf{P}(t), P_E(t), x) - \kappa_i \right\} m^i_t(dx)\,dt \\
&\quad - \int_0^T \int_{\overline{\mathcal{O}}_i} K_i e^{-(r + \gamma_i)t} \mu^i(dt, dx) 
- \int_0^T e^{-r t} \Delta_i(t, P^i(t))\,dt,
\end{split}
\end{equation}
where \( \Delta_i(t, P) \coloneqq \int_1^P D_i(t, z)\,dz \) denotes the consumer surplus function of sector \( i \) at time \( t \), evaluated up to price level \( P \geq 0 \).

In analogy to Definition \ref{def:saddle_point}, we say that a triple $(\mathbf{P}^*, \mathbf{m}^*, \boldsymbol{\mu}^*)\in \cP^q\times\cR(\mathbf{m}_0)$ is a saddle point of $\cL$ on $\cP^q\times\cR(\mathbf{m}_0)$ if and only if it satisfies the Minimax equality:
\begin{equation}\label{eq:MFG_minimax}
\cL(\mathbf{P}^*,\mathbf{m}^*,\boldsymbol{\mu}^*)=\max_{\mathbf{m},\boldsymbol{\mu}}\min_{\mathbf{P}}\cL(\mathbf{P},\mathbf{m},\boldsymbol{\mu})=\min_{\mathbf{P}}\max_{\mathbf{m},\boldsymbol{\mu}}\cL(\mathbf{P},\mathbf{m},\boldsymbol{\mu}).
\end{equation}

We first establish the connection between the notion of a saddle point of $\mathcal{L}$ and a linear programming MFG Nash equilibrium. The proof is provided in Appendix~\ref{app:2}. 

\begin{lemma}\label{lem:MFG_saddle_point}
Suppose the production, labor cost, and net demand functions satisfy Assumptions \ref{ass:prod_function}(b), \ref{ass:wage_function}, \ref{ass:wage_function2}, and \ref{ass:demand_uniform}. Furthermore, assume that the production functions \( F_i \) are strictly concave and the labor cost functions \( W_i \) are strictly convex and differentiable.  

Let \( (\mathbf{P}^*,\mathbf{m}^*,\boldsymbol{\mu}^*)\in\cP^+_q\times\cR(\mathbf{m}_0) \) be a saddle point of \( \cL \). Additionally, let \( (\mathbf{q}^{i,*},E^{i,*},L^{i,*}) \) be a solution to the instantaneous profit maximization problem \eqref{eq:max_inst_profit}.  

Then, the tuple \( (\mathbf{P}^*,\mathbf{m}^*,\boldsymbol{\mu}^*,\cQ^*,\mathbf{E}^*,\mathbf{L}^*) \) constitutes a linear programming MFG Nash equilibrium.
\end{lemma}

As in the static version of the model, the existence of a linear programming MFG Nash equilibrium can be established under either the market viability Assumption \eqref{ass:participation} or, alternatively, by imposing a growth condition on the net demand functions that ensures external supply is sufficiently unbounded.

\begin{assumption}[Growth condition on net demand functions]\label{ass:supply_MFG}
For each \( i = 1, \dots, N \), the net demand function \( D_i \) satisfies the following asymptotic condition: 
\[
\lim_{P \to \infty} \frac{D_i(t,P)}{P^{q-1}} = -\infty,
\]
for some exponent \( q > 1 \), uniformly for all \( t \in [0,T] \).
\end{assumption}

This condition ensures that external supply grows faster than \( P^{q-1} \) as prices tend to infinity. 

Together with additional conditions on the convex conjugates \( W^*_i \), Assumptions \ref{ass:participation} and \ref{ass:supply_MFG} provide sufficient conditions for the existence of a linear programming MFG Nash equilibrium, as summarized in the following theorem.

\begin{theorem}\label{th:MFG_equilibrium}
Assume that the production, labor cost, and net demand functions satisfy Assumptions \ref{ass:prod_function} (b), \ref{ass:wage_function}, \ref{ass:wage_function2}, and \ref{ass:demand_uniform}. Furthermore, suppose the production functions \( F_i \) are strictly concave, and the labor cost functions \( W_i \) are strictly convex and differentiable. Additionally, assume that Assumptions \ref{ass:SDE} holds, and that for each \( i=1,\dots,N \), the initial distribution \( m^i_0 \) satisfies  
\[
\int_{\cO_i} (1 + x^{p'}) \, m_0(dx) < \infty, \quad \text{for some $ p' > p$}.
\]     
Finally, let $q>1$ and suppose that one of the following conditions holds:
\begin{enumerate}
\item Assumption \ref{ass:participation} is satisfied, and there exist positive constants \( C_1, C_2\) and  $y_0$ such that  for all $i=1,\dots,N$, 
\[
C_1 y^q \leq \abs{W^*_i(y)} \leq C_2 y^q, \quad \text{for } y \geq y_0.
\]
\item The net demand functions \( D_i \) satisfy Assumption \ref{ass:supply_MFG} and \( W^*_i(y) = O(y^q) \) for all $i=1,\dots,N$. 
\end{enumerate}

Then, there exists a linear programming MFG Nash equilibrium \( (\mathbf{P}^*,\mathbf{m}^*,\boldsymbol{\mu}^*) \in \cP^+_q \times \cR(\mathbf{m}_0) \). Moreover, the equilibrium price system is unique up to sets of measure zero.
\end{theorem}

\begin{proof}
By Lemma \ref{lem:MFG_saddle_point}, the existence of a linear programming MFG Nash equilibrium reduces to establishing the existence of a saddle point for $\cL$. We will establish the existence of a saddle point on $\cP^+_q \times \cR(\mathbf{m}_0)$. To prove this, we follow the approach outlined in Theorem \ref{th:Nash_equilibrium} and verify that the Lagrangian function defined in \eqref{eq:MFG_lagrangian} satisfies the criteria of Proposition VI.2.3 in \cite{Ekeland99}. The proof is structured in a series of steps.

\step[Lower semi-continuity and convexity of $\cL$ in $\mathbf{P}$]
By Assumption \ref{ass:demand_uniform}, the Lagrangian $\cL(\mathbf{P},\mathbf{m},\boldsymbol{\mu})$ is (strictly) convex with respect to $\mathbf{P}$. To establish its lower semicontinuity in the price variable, consider a sequence $(\mathbf{P}^n)_{n\geq 1}$ converging to a limit $\mathbf{P}$ in $L^q([0,T])$. Without loss of generality, we may assume (by selecting a subsequence if necessary) that this convergence occurs almost everywhere on $[0,T]$. By Fatou's lemma, it follows that, for each $i=1,\dots,N$:
\[
\liminf_{n\to \infty} \int_0^T\int_0^\infty  \overline\Pi_i(\mathbf{P}^n(t),P_E(t),x)\, m^i_t(dx)dt\geq 
\int_0^T\int_0^\infty  e^{-rt}\overline\Pi_i(\mathbf{P}(t),P_E(t),x)\, m^i_t(dx)dt.
\]
Next, Assumption \ref{ass:demand_uniform} ensures that, for every $\varepsilon>0$, there exist constants $C>0$ such that:
\[
\Delta_i(t,P)\leq\varepsilon P + C,\quad \text{for all $P>0$}.
\]
uniformly over $t\in [0,T]$ and for each $i=1,\dots,N$. Applying Fatou's lemma again, we obtain:
$$
\liminf_{n\to \infty}\int^T_0 e^{-r t}\{-\Delta_i(t,P^n_i(t))\} dt \geq  \int^T_0 e^{-r t}\{-\Delta_i(t,P_i(t))\}dt.
$$
Combining these results, we deduce:
\[
\liminf_{n\to \infty}\cL(\mathbf{P}^n,\mathbf{m},\boldsymbol{\mu})\geq \cL(\mathbf{P},\mathbf{m},\boldsymbol{\mu}),\quad \text{for any $(\mathbf{m},\boldsymbol{\mu})\in\cR(\mathbf{m}_0)$},
\]
Thus, we conclude that the Lagrangian $\cL$ is lower semi-continuous with respect to the price variable. 

\step[Upper semi-continuity and concavity of $\cL$ in $(\mathbf{m},\boldsymbol{\mu})$]

We now verify that the Lagrangian $\cL$ satisfies the required properties as a function of $(\mathbf{m},\boldsymbol{\mu})\in\cR(\mathbf{m}_0)$. By linearity, it is straightforward to observe that $\cL$ is concave with respect to $(\mathbf{m},\boldsymbol{\mu})$.

To establish that $\cL$ is upper semi-continuous in $(\mathbf{m},\boldsymbol{\mu})$, let $(\mathbf{m}^n,\boldsymbol{\mu}^n)_{n\geq 1}\subset \cR(\mathbf{m}_0)$  be a sequence converging to  $(\mathbf{m},\boldsymbol{\mu})\in\cR(\mathbf{m}_0)$ under the topology $\tau^N_p\otimes\tau^N_p$. Moreover, consider sequences of bounded continuous functions $(P^{m}_E)_{m\geq 1}$ and $(\mathbf{P}^m)_{m\geq 1}$ that approximate $P_E$ and $P$ in $L^{q}([0,T])$.

We now decompose \(\cL\), explicitly highlighting its dependence on the carbon price \( P_E \):  
\begin{align}\label{estimL}
\cL(\mathbf{P},P_E,\mathbf{m}^n,\boldsymbol{\mu}^n) &= \cL(\mathbf{P},P_E,\mathbf{m}^n,\boldsymbol{\mu}^n)-\cL(\mathbf{P}^m,P^{m}_E,{\mathbf{m}}^n,{\boldsymbol{\mu}}^n)\\
&\quad+\cL(\mathbf{P}^m,P^{m}_E,\mathbf{m}^n,\boldsymbol{\mu}^n). \label{estimL2}
\end{align}  

To estimate the first term, we utilize the explicit form of the maximized profit function:  
\[
\overline\Pi_i(\mathbf{P},P_E,x) = W^*_i\left(\widetilde\Pi_i(\mathbf{P},P_E) - x\right).
\]  
Applying the convexity of the convex conjugate \( W^*_i \) and the relative compactness of the sequence \( (m^{i,n})_{n\geq 1} \) in the \(\tau^i_p\)-topology (as established in Lemma \ref{lem:compactness}), we obtain:  
\begin{align*}
&\abs*{\cL(\mathbf{P},P_E,\mathbf{m}^n,\boldsymbol{\mu}^n)-\cL(\mathbf{P}^m,P^{m}_E,{\mathbf{m}}^n,{\boldsymbol{\mu}}^n)} \\
&\leq \sum^N_{i=1}\int^T_0\int_{\overline\cO_i}\abs*{W^*_i\left(\widetilde\Pi_i(\mathbf{P}(t),P_E(t)) - x\right) - W^*_i\left(\widetilde\Pi_i(\mathbf{P}^m(t),P^{m}_E(t)) - x\right)}\,m^{i,n}_t(dx)dt \\
&\leq \sum^N_{i=1}\int^T_0\abs*{W^*_i\left(\widetilde\Pi_i(\mathbf{P}(t),P_E(t))\right) - W^*_i\left(\widetilde\Pi_i(\mathbf{P}^m(t),P^{m}_E(t))\right)} \int_{\overline\cO_i} m^{i,n}_t(dx)dt \\
&\leq C\sum^N_{i=1}\int^T_0\abs*{W^*_i\left(\widetilde\Pi_i(\mathbf{P}(t),P_E(t))\right) - W^*_i\left(\widetilde\Pi_i(\mathbf{P}^m(t),P^{m}_E(t))\right)}\,dt,
\end{align*}  
for some constant \( C>0 \) independent of \( n \).  

By the definition of \( \widetilde \Pi_i \) and Assumption \ref{ass:prod_function} (b), we have \( \widetilde \Pi_i(\mathbf{P}) \leq c P_i \), and since \( \widetilde \Pi_i \) is convex as the pointwise supremum of linear functions, it follows that \( \widetilde \Pi_i \) is continuous. Similarly, \( W^*_i \) is a proper convex function and hence continuous.  

Thus, by extracting a subsequence if necessary, we may assume that \( (P^{m}_E)_{m\geq 1} \) and \( (\mathbf{P}^m)_{m\geq 1} \) converge to their respective limits almost everywhere on \( [0,T] \). Consequently,  
\begin{equation}\label{eq:as_limit}
\abs*{W^*_i\left(\widetilde\Pi_i(\mathbf{P}(t),P_E(t))\right) - W^*_i\left(\widetilde\Pi_i(\mathbf{P}^m(t),P^{m}_E(t))\right)} \to 0, \quad \text{a.e. on \( [0,T] \).}
\end{equation} 

Moreover, under both conditions (1) and (2) of the theorem, the upper bound \( \widetilde \Pi_i(\mathbf{P}) \leq c P_i \), together with the convergence of the sequences \( (\mathbf{P}^m) \) and \( (P^{m}_E) \) in \( L^q([0,T]) \), implies that the sequence of functions
\[
\left(W^*_i(\widetilde\Pi_i(\mathbf{P}^m,P^{m}_E)\right)_{m\geq 1}
\]  
is uniformly integrable. Consequently, the entire expression in \eqref{eq:as_limit} is also uniformly integrable.  

Applying Vitali's convergence theorem, we deduce that  
\begin{equation}\label{eq:Delta_limit}
\abs*{\cL(\mathbf{P},P_E,\mathbf{m}^n,\boldsymbol{\mu}^n)-\cL(\mathbf{P}^m,P^{m}_E,{\mathbf{m}}^n,{\boldsymbol{\mu}}^n)} \to 0, \quad \text{as } m\to\infty,
\end{equation}  
uniformly over \( n\geq 1 \).  

Following the same reasoning, we obtain  
\begin{equation}\label{eq:Delta_limit2}
\cL(\mathbf{P}^m,P^{m}_E,\mathbf{m},\boldsymbol{\mu})\to\cL(\mathbf{P},P_E,\mathbf{m},\boldsymbol{\mu}), \quad \text{as } m\to\infty.
\end{equation}

Regarding the term in \eqref{estimL2}, using the definition of convergence in the topology \( \tau_p \) and the continuity of the sequences \( (P^{m}_E)_{m \geq 1} \) and \( (\mathbf{P}^m)_{m \geq 1} \), we obtain:  
\begin{align*}
\lim_{n\to\infty}\cL(\mathbf{P}^m,P^{m}_E,\mathbf{m}^n,\boldsymbol{\mu}^n) &= \sum^N_{i=1}\int_0^T\int_{\overline\cO_i} e^{-rt} \left\{W^*_i\left(\widetilde\Pi_i(\mathbf{P}^m(t),P^{m}_E(t))-x\right) -\kappa_i\right\}\,m^i_t(dx) dt \\
&\quad -\sum^N_{i=1} \int_0^T\int_{\overline\cO_i} K_i e^{-(r+\gamma_i)t}\,\mu^i(dt,dx) - \int^T_0 e^{-r t}\Delta_i(t,P^{m}_i(t))\,dt \\
&= \cL(\mathbf{P}^m,P^{m}_E,\mathbf{m},\boldsymbol{\mu}).
\end{align*}  
Combining this result with the convergence established in Equations \eqref{eq:Delta_limit}–\eqref{eq:Delta_limit2} and substituting back into Equation \eqref{estimL}, we conclude:  
\[
\lim_{n\to\infty}\cL(\mathbf{P},P_E,\mathbf{m}^n,\boldsymbol{\mu}^n) = \cL(\mathbf{P},P_E,\mathbf{m},\boldsymbol{\mu}),\quad \text{for any } \mathbf{P} \in \cP_q.
\]  
Thus, we have established that the Lagrangian \( \cL \) is continuous with respect to \( (\mathbf{m},\boldsymbol{\mu}) \in \cR(\mathbf{m}_0) \) and, in particular, it is lower semi-continuous.

\step[Existence of a saddle point of $\cL$ on $\cP_q\times\cR(\mathbf{m}_0)$]

It is straightforward to verify that $\cR(\boldsymbol{m}_0)$ is convex within $\cV^N_p\times\cM^N_p$. Furthermore, by Lemma \ref{lem:compactness}, $\cR(\boldsymbol{m}_0)$ is the product of weakly compact subsets of $\cV^N_p\times\cM^N_p$. As a result, $\cR(\boldsymbol{m}_0)$ is weakly compact in $\cV^N_p\times\cM^N_p$ under the topology  $\tau^N_p\otimes\tau^N_p$. 

By applying Proposition VI.2.3 of \cite{Ekeland99}, we deduce: 
\begin{equation}\label{eq:inf_max}
\max_{\mathbf{m},\boldsymbol{\mu} \in\mathcal R(\mathbf m_0)} \inf_{\mathbf P\in \cP_q} \mathcal L(\mathbf P, \mathbf m, \boldsymbol{\mu}) = \inf_{\mathbf P\in \cP_q}\max_{\mathbf{m},\boldsymbol{\mu} \in \mathcal R(\mathbf m_0)}  \mathcal L(\mathbf P, \mathbf m, \boldsymbol{\mu}) 
\end{equation}

To complete the proof of the existence of the saddle point, by Proposition VI.1.2 of \cite{Ekeland99}, it suffices to show that the infimum on the right-hand side is attained.  

To this end, we establish that the maximized Lagrangian is coercive in the \( L^q([0,T]) \)-norm:
\begin{equation}\label{eq:Lq_coercitivity}
\max_{\mathbf{m},\boldsymbol{\mu}}  \mathcal L(\mathbf P, \mathbf m, \boldsymbol{\mu})\to\infty,\quad\text{as $\norm{\mathbf{P}}_{L^q}\to +\infty$}.  
\end{equation}

Assume first that condition (1) of the theorem holds. By Assumption \ref{ass:participation}, we can fix a pair \( (\overline{\cQ},\overline{\mathbf{E}}) \in \mathbb{R}^{N(N+1)}_{\geq 0} \) and a price vector \( \mathbf{P} \in \mathbb{R}^N_{>0} \) such that  
\[
\overline{P}_i F_i(\overline{\mathbf{q}}_i, \overline{E}_i, 1) - \sum_{j=1}^N \overline{P}_j \overline{q}_{ij} > 0, \quad \text{for all } i = 1, \dots, N.
\]

Next, define a function \( k:[0,T] \to \{1,\dots,N\} \) such that  
\begin{equation}\label{eq:price_k}
\frac{P_{k(t)}(t)}{\overline P_{k(t)}} \geq \frac{P_i(t)}{\overline P_i},\quad \text{for all \( t \in [0,T] \) and \( i=1,\dots,N \)}.
\end{equation}

Using the definition of the maximum and the explicit formula for the maximized instantaneous profit functions, we obtain:
\begin{equation}\label{eq:maxL_estm}
\begin{split}
\max_{\mathbf{m},\boldsymbol{\mu}}\mathcal L(\mathbf P, \mathbf m, \boldsymbol{\mu})&\geq\int_0^T\int_{\overline\cO_{k(t)}}  e^{-rt}\left\{W^*_{k(t)}\left(\widetilde\Pi_{k(t)}(\mathbf{P}(t))-x\right)-\kappa_{k(t)}\right\} m^{k(t)}_0(dx)dt\\
&\quad -\sum^N_{i=1}\int^T_0 e^{-r t}\Delta_{i}(t,P_i(t))\,dt,
\end{split}
\end{equation}
where, for simplicity of notation, we suppress the dependence of \( \widetilde\Pi_{k(t)} \) on \( P_E(t) \).

The optimized profit function satisfies  
\begin{align*}
\widetilde\Pi_{k(t)}(\mathbf{P}(t)) &\geq  P_{k(t)}(t) F_{k(t)}(\mathbf{\overline q}_{k(t)}, \overline E_{k(t)},1) - \sum_{j=1}^N P_j(t) \overline q_{j} - P_E(t) \overline E_{k(t)}  \\
&\geq \frac{P_{k(t)}(t)}{\overline P_{k(t)}}\left\{\overline P_{k(t)} F_{k(t)}(\mathbf{\overline q}_{k(t)}, \overline E_{k(t)},1) - \sum_{j=1}^N \overline P_j \overline q^{j}  \right\}- P_E(t) \overline E_{k(t)},  
\end{align*}
for all \( t\in [0,T] \).  

Moreover, in view of \eqref{eq:price_k}, there exists a constant \( c>0 \) such that  
\[
P_{k(t)}(t) \geq c\|\mathbf P(t)\|_q,\quad\text{for all $t\in[0,T]$}.
\]
It follows that  
\[
\widetilde\Pi_{k(t)}(\mathbf{P}(t)) \geq c' \|\mathbf P(t)\|_q - P_E(t) \overline E_{k(t)},\quad \text{for all \( t\in [0,T] \)}
\]
for some constant \( c' > 0 \).
 
This, together with condition (1) of the theorem, leads to the following bound for the first term on the right-hand side of \eqref{eq:maxL_estm}, for sufficiently large \( \|\mathbf P\|_{L^q} \):  
\begin{align*}
&\int_0^T\int_{\overline\cO_{k(t)}}  e^{-rt}\left\{W^*_{k(t)}\left(\widetilde\Pi_{k(t)}(\mathbf{P}(t))-x\right)-\kappa_{k(t)}\right\} m^{k(t)}_0(dx)dt\\
&\geq C +\int_0^T\int_{\overline\cO_{k(t)}}  e^{-rt} W^*_{k(t)}\left( c' \|\mathbf P(t)\|_q - P_E(t) \overline E_{k(t)}-x\right) m^{k(t)}_0(dx)dt\\
&\geq C + \varepsilon\int_0^T\int_{\overline\cO_{k(t)}}  e^{-rt} \left( c' \|\mathbf P(t)\|_q - P_E(t) \overline E_{k(t)}-x\right)^{q}_+ m^{k(t)}_0(dx)dt\\
&\geq C + \varepsilon\norm{\mathbf P}^q_{L^q}
\end{align*}
where the constant \( C \) and the strictly positive constant \( \varepsilon \) may change from line to line. To derive the last inequality, we use the fact that, without loss of generality, we may assume  
\[
\int_{\overline{\cO}_j} m^j_0(dx) > K, \quad \text{for all } j = 1, \dots, N,
\]
for some constant \( K > 0 \).

For the second term on the right-hand side of \eqref{eq:maxL_estm}, by Assumption \ref{ass:demand_uniform}, we know that for every \( \varepsilon'>0 \), there exists a constant \( C'>0 \) such that:
\[
\int^T_0 e^{-rt}\Delta_i(t,P^i_t)\,dt\leq\varepsilon'\norm{P^i}_{L^1} + C',\quad \text{for all \( i=1,\dots,N \)}.
\]
Applying this bound and using Hölder's inequality, we obtain:
\begin{equation}\label{eq:demand_holder}
-\sum^N_{i=1}\int^T_0 e^{-rt}\Delta_i(t,P_i(t))\,dt\geq- \varepsilon'\norm{\mathbf{P}}_{L^{q}}-C'N.  
\end{equation}
Combining this with the previous estimate, we conclude that there exist constants \( C'' \), as well as \( c'', \varepsilon''>0 \), such that:
\[
\max_{\mathbf{m},\boldsymbol{\mu}}  \mathcal L(\mathbf P, \mathbf m, \boldsymbol{\mu}) \geq C'' + \varepsilon'' \|\mathbf P\|_{L^q}^{q},\quad \text{for } \|\mathbf P\|_{L^q} \geq c''.
\]
This confirms that the coercivity condition \eqref{eq:Lq_coercitivity} is satisfied under Assumption (1) of the theorem.

Assume now that condition (2) of the theorem holds. By Assumption~\ref{ass:supply_MFG}, for every $\varepsilon > 0$, there exist constants $C \in \mathbb{R}$ and $c > 0$ such that  
\[
\Delta_i(t, P) \leq C - \varepsilon P^q \quad \text{for all } t \in [0,T] \text{ and } P \geq c.
\]
Consequently, for any $\varepsilon' > 0$, there exist constants $C' \in \mathbb{R}$ and $c' > 0$ such that  
\[
\sum^N_{i=1} \int_0^T e^{-rt} \Delta_i(t, P_i(t)) \, dt < C' - \varepsilon' \|\mathbf{P}\|^q_{L^q}, \quad \text{whenever } \|\mathbf{P}\|_{L^q} \geq c'.
\]
Using this, we obtain the estimate:
\begin{equation}\label{eq:maxL_estm2}
\begin{split}
\max_{\mathbf{m},\boldsymbol{\mu}} \mathcal L(\mathbf{P}, \mathbf{m}, \boldsymbol{\mu}) 
&\geq \sum^N_{i=1} \int_0^T \int_{\overline{\cO}_i}  e^{-rt} \left\{ W^*_{i} \left(\widetilde{\Pi}_{i}(\mathbf{P}(t)) - x \right) - \kappa_{i} \right\} m^{i}_0(dx)dt \\
&\quad + \varepsilon'\|\mathbf{P}\|^q_{L^q}-C'.
\end{split}
\end{equation}

By condition (2) of the theorem, we have that \( W^*_i(y) = O(y^q) \) for all \( i = 1, \dots, N \). Combining this with the upper bound \( \widetilde{\Pi}_i(\mathbf{P}) \leq c P_i \), we obtain that for sufficiently large \( \|\mathbf{P}(t)\|_q \),
\begin{align*}
\int_{\overline{\cO}_i} W^*_{i}\left(\widetilde{\Pi}_{i}(\mathbf{P}(t)) - x \right)\, m^i_0(dx) 
&\leq \varepsilon'' \int_{\overline{\cO}_i} \left(\widetilde{\Pi}_{i}(\mathbf{P}(t)) - x \right)^q_+ \,m^i_0(dx) \\
&\leq \varepsilon'' \int_{\overline{\cO}_i} \Ind_{x\leq\widetilde{\Pi}_i(\mathbf{P}(t))} \left(\widetilde{\Pi}_{i}(\mathbf{P}(t))^q + x^q \right)\,m^i_0(dx) \\
&\leq \varepsilon'' P_i(t)^q,
\end{align*}
for some constant \( \varepsilon'' > 0 \), which may vary from line to line.

By integrating both sides with respect to \( t \in [0,T] \) and summing over \( i = 1, \dots, N \), we obtain that there exists a constant \( c'' > 0 \) such that:
\[
\sum^N_{i=1} \int_0^T \int_{\overline{\cO}_i} W^*_{i} \left(\widetilde{\Pi}_{i}(\mathbf{P}(t)) - x \right)\,m^{i}_0(dx)dt
\leq \varepsilon'' \|\mathbf{P}\|^q_{L^q}, \quad \text{for } \|\mathbf{P}\|_{L^q} \geq c''.
\]
This immediately implies that:
\[
\sum^N_{i=1} \int_0^T \int_{\overline{\cO}_i} W^*_{i} \left(\widetilde{\Pi}_{i}(\mathbf{P}(t)) - x \right)\,m^{i}_0(dx)dt = O(\|\mathbf{P}\|^q_{L^q}).
\]
Substituting this into \eqref{eq:maxL_estm2} establishes the coercivity condition \eqref{eq:Lq_coercitivity}.

Finally, from the coercivity of the maximized Lagrangian in the \( L^{q}([0,T]) \)-norm, it follows that if there exists a price vector function \( \mathbf{P}^* \) such that the infimum on the right-hand side of Equation \eqref{eq:inf_max} is attained, then \( \mathbf{P}^* \in \mathcal{K} \), where \( \mathcal{K} \) is a closed, bounded, and convex subset of \( \mathcal{P}_q \) defined by
\[
\mathcal{K} = \left\{ \mathbf{P} \in\cP_q : \|\mathbf{P}\|_{L^q} \leq K \right\},
\]
for some constant \( K > 0 \). Since \( L^q([0,T],\R^N)\) is a reflexive Banach space for \( 1 < q < \infty \), it follows from Kakutani’s theorem (see \cite[Theorem V.4.2]{Conway1985}) that \( \mathcal{K} \) is weakly compact, ensuring that the infimum in \eqref{eq:inf_max} is attained.

This completes the proof of the existence of a saddle point \( (\mathbf{P}^*, \mathbf{m}^*, \boldsymbol{\mu}^*) \in \mathcal{P}_q \times \mathcal{R}(\mathbf{m}_0) \).

\step[Almost sure positivity and uniqueness of $\mathbf{P}^*$]

To apply Lemma \ref{lem:MFG_saddle_point}, it remains to verify that \( \mathbf{P}^* \in \cP^+_q \). 

By definition of a saddle point, \( \mathbf{P}^* \) minimizes the Lagrangian \( \mathcal{L}(\mathbf{P}, \mathbf{m}^*,\boldsymbol{\mu}^*) \). Suppose, for contradiction, that there exists a measurable set \( B \subseteq [0,T] \) with \( \operatorname{Leb}(B) > 0 \) and some \( i \in \{1, \dots, N\} \) such that \( P^{*}_i(t) = 0 \) for all \( t \in B \). 

Consider the perturbed price function \( \mathbf{P}^{h} \in \cP_q \) for some \( h \in (0,1) \), defined as:
\[
P^{h}_j(t) = P^{*}_j(t), \quad \forall j \neq i, \quad \forall t \in [0,T],
\]
and
\[
P^{h}_i(t) =
\begin{cases}
P^{*}_i(t), & t \in [0,T] \setminus B, \\
h, & t \in B.
\end{cases}
\]

{By the mean value theorem, there exists a function \( \mathbf{v}^h \in \cP_q \) such that
\begin{align*}
v^{h}_j(t) &= P^{*}_j(t), \quad \forall j \neq i, \quad \forall t \in [0,T], \\
v^{h}_i(t) &= P^{*}_i(t), \quad \forall t \in [0,T] \setminus B, \\
v^{h}_i(t) &\in (0,h), \quad \forall t \in B,
\end{align*}
and satisfying
\begin{equation}\label{eq:mean_value}
\mathcal{L}(\mathbf{P}^h,\mathbf{m}^*,\boldsymbol{\mu}^*) - \mathcal{L}(\mathbf{P}^*, \mathbf{m}^*,\boldsymbol{\mu}^*) = h \int_B \frac{\delta}{\delta P_i}\mathcal{L}(\mathbf{v}^h, \mathbf{m}^*,\boldsymbol{\mu}^*)(t)\, dt,
\end{equation}
where \( \frac{\delta}{\delta P_i} \) denotes the Gâteaux derivative of \( \mathcal{L} \) with respect to \( P_i \). In particular,
\[
\frac{\delta}{\delta P_i}\mathcal{L}(\mathbf{v}^h, \mathbf{m}^*,\boldsymbol{\mu}^*) (t)= \sum_{j=1}^{N} \frac{\partial}{\partial P_i} \int_{\overline{\cO}_j}  e^{-rt} \overline{\Pi}_j(\mathbf{v}^h(t), x) \, m^{j,*}_t(dx) - e^{-r t} D_i(t, v^{h}_i(t)).
\]
For notational simplicity, we suppress the dependence on the carbon price.}

Following the same argument as in the proof of Lemma~\ref{lem:MFG_saddle_point}, we can exchange differentiation and integration over \( m^{j,*}_t \), yielding
\begin{equation}\label{eq:L_derivative}
\begin{split}
\int_B \frac{\delta}{\delta P_i}\mathcal{L}(\mathbf{v}^h, \mathbf{m}^*,\boldsymbol{\mu}^*)(t)\, dt &= \sum_{j=1}^{N} \int_B \int_{\overline{\cO}_j}  e^{-rt} \frac{\partial}{\partial P_i} \overline{\Pi}_j(\mathbf{v}^h(t), x) \, m^{j,*}_t(dx) \, dt \\
&\quad - \int_B e^{-r t} D_i(t, v^{h}_i(t)) \, dt.
\end{split}
\end{equation}

For the first term, we use the explicit formulas for the maximized instantaneous profit and the optimal labor choice from Proposition~\ref{prop:max_inst_profit}, which give
\begin{equation}\label{eq:bound_der}
\frac{\partial}{\partial P_i} \overline{\Pi}_j(\mathbf{v}^h(t), x) = (W^*_j)'\left(\widetilde{\Pi}_j(\mathbf{v}^h(t)) - x\right) \, \partial_{P_i} \widetilde{\Pi}_j(\mathbf{v}^h(t)).
\end{equation}
{By Danskin’s theorem, the partial derivative of $\widetilde{\Pi}_j$ with respect to $P_i$ is
\[
\frac{\partial}{\partial P_i} \widetilde{\Pi}_j(\mathbf{v}^h(t)) = \Ind_{i=j} \left\{ F_j\left(\widetilde{\mathbf{q}}^{*}_j(\mathbf{v}^h(t)), \widetilde{E}^{*}_j(\mathbf{v}^h(t)), 1\right) - q^*_{ji}(\mathbf{v}^h(t)) \right\} \leq c \Ind_{j=i},
\]
where the inequality follows from Assumption~\ref{ass:prod_function}.}  
Substituting this into~\eqref{eq:bound_der}, and using both the upper bound \( \widetilde{\Pi}_i(\mathbf{P}) \leq c P_i \) and the convexity of \( W^*_i \), we obtain
\begin{align*}
\frac{\partial}{\partial P_i} \overline{\Pi}_j(\mathbf{v}^h(t), x) \leq c \Ind_{j=i} \, (W^*_j)'\left(c v^{h}_j(t) - x\right).
\end{align*}
Integrating both sides and summing over \( j=1,\dots,N \), we get
\begin{align*}
\sum^N_{j=1} \int_B \int_{\overline{\cO}_j}  e^{-rt} \frac{\partial}{\partial P_i} \overline{\Pi}_j(\mathbf{v}^h(t), x) \, m^{j,*}_t(dx) dt
&\leq c\int_B\int_{\overline{\cO}_i} e^{-rt}  (W^*_i)'\left(c v^{h}_i(t) - x\right)\, m^{i,*}_t(dx) dt.
\end{align*}

Since \( v^{h}_i(t) \to 0 \) as \( h \to 0 \) for all \( t \in B \), and given that \( (W^*_i)' \) has the same support as \( W^*_i \), meaning it vanishes on the negative real line, it follows that
\[
\lim_{h\to 0} (W^*_i)'\left(c v^{h}_i(t) - x\right) = 0, \quad \text{for all } (t,x) \in B\times\overline{\cO}_i.
\]
Furthermore, since \( (W^*_i)' \) is increasing, we have the upper bound
\[
(W^*_i)'\left(c v^{h}_i(t) - x\right) \leq (W^*_i)'\left(c v^{h}_i(t)\right) \leq  (W^*_i)'\left(c\right),
\]
where we used \( v^{h}_i(t) \leq 1 \) for all \( t \in B \) and uniformly over \( h \in (0,1) \).

Applying the dominated convergence theorem, we conclude that
\[
\sum^N_{j=1} \int_B \int_{\overline{\cO}_j}  e^{-rt} \frac{\partial}{\partial P_i} \overline{\Pi}_j(\mathbf{v}^h(t), x) \, m^{j,*}_t(dx) dt = o(1).
\]

Plugging this bound into Equation \eqref{eq:L_derivative}, we obtain  
\begin{align*}
\int_B\frac{\delta}{\delta P_i(t)} \mathcal{L}(\mathbf{v}^h, \mathbf{m}^*,\boldsymbol{\mu}^*)\,dt 
= o(1) - \int_B e^{-r t} D_i(t, v^{h}_i(t)) \, dt.
\end{align*}  

By Assumption \ref{ass:demand_uniform},  
\[
\lim_{h\to 0} D_i(t, v^{h}_i(t)) = \infty, \quad \text{uniformly for } t \in B.
\]
Thus, by uniform convergence, for sufficiently small \( h \), we have  
\[
\int_B\frac{\delta}{\delta P_i(t)} \mathcal{L}(\mathbf{v}^h, \mathbf{m}^*,\boldsymbol{\mu}^*)\,dt < 0.
\]
This, combined with Equation \eqref{eq:mean_value}, implies  
\[
\mathcal{L}(\mathbf{P}^h,\mathbf{m}^*,\boldsymbol{\mu}^*) < \mathcal{L}(\mathbf{P}^*, \mathbf{m}^*,\boldsymbol{\mu}^*),
\]
contradicting the assumption that \( \mathbf{P}^* \) minimizes \( \mathcal{L} \). We thus conclude that \( \mathbf{P}^* \in \cP^+_q \), as required.

This concludes the proof of the existence of a linear programming MFG Nash equilibrium. The uniqueness of the corresponding price system (up to a set of measure zero) follows directly from the strict convexity of $\cL(\mathbf{P},\mathbf{m},\boldsymbol{\mu})$ with respect to $\mathbf{P}$.
\end{proof}

\begin{remark}[Connection to Mean-Field Type Control Problems]
As in the static model (see Remark~\ref{rem:welfare}), the equilibrium resulting from decentralized, non-cooperative decision-making in a mean-field game (right-hand side of~\eqref{eq:MFG_minimax}) coincides with the allocation that a social planner would implement to maximize the aggregate surplus of the economy (left-hand side of~\eqref{eq:MFG_minimax}). 

Because of the mean-field structure of the model, the social planner formulation corresponds to a mean-field type control (MFC) problem, in which the planner optimally controls the distribution of agents to maximize a welfare functional and achieve Pareto efficient outcomes (see \cite[Ch.~4]{Bensoussan2013}).
\end{remark}

\section{Numerical illustration}\label{sec:4}

In this section, we present numerical simulations based on CES production functions, CIR dynamics for labor cost processes, and power-type functional forms for demand and labor cost. We consider two stylized models.

We first analyze a single-sector economy consisting solely of a brown sector, which employs labor and emissions as inputs with equal CES share parameters. This setting allows us to study the effects of a carbon price on (i) capacity evolution, (ii) substitution dynamics — specifically, how emissions are replaced by labor in output production — and (iii) carbon cost pass-through, i.e., the extent to which the carbon price is reflected in final prices passed on to consumers (see \cite{neuhoff2019carbon}).

We then turn to a multi-sector economy composed of three sectors: Brown, Green, and Manufacturing. In this case, the CES input share parameters are set to reflect an input-output structure where Manufacturing requires both Brown and Green goods as inputs, and both Brown and Green sectors consume a fraction of Manufacturing output in their production processes. This setup enables us to study the propagation of carbon pricing through the value chain.

To isolate the effects of substitution and propagation, we fix the parameters of the CIR labor cost process uniformly across all sectors (see Remark~\ref{rem:CIR}), setting \( \alpha = 0.1 \), \( \theta = 22.5 \), and \( \sigma = 8 \).

We further assume zero liquidation cost at default, allowing us to focus on sector capacities, as well as the densities of the killed labor cost processes. 

The linear MFG Nash equilibrium is computed using the fictitious play algorithm proposed by \cite{aid2021entry} and \cite{dumitrescu2023linear}, combined with explicit expressions for optimal input choices derived in Propositions~\ref{prop:max_inst_profit} and~\ref{ces.ex}. For each best response in Problem~\eqref{prob:MFG_profit}, we discretize the Fokker-Planck inequality using an implicit scheme, and evaluate the optimization functional over discretized measures. The resulting optimization problem is solved using the Gurobi library (\url{www.gurobi.com}).

\subsection{Example 1: Impact of carbon pricing on a single brown sector}

In the single-sector model, we consider a brown sector that requires 0.5 units of labor and 0.5 units of emissions to produce one unit of output, with a CES substitution parameter \( \rho = 0.5 \), corresponding to a moderately high elasticity of substitution. The carbon price increases deterministically from \$1 to \$30 over the period 2025--2030.

Figure~\ref{1sector.capacity} (left panel) displays the evolution of the sector's total capacity under two scenarios: with and without carbon pricing. The right panel shows the corresponding evolution of the density of labor costs among active firms. The introduction of a carbon price clearly reduces capacity, which declines significantly faster than in the no-tax scenario. However, due to the deterministic nature of the policy, firms are able to anticipate the shock and gradually exit the market according to their individual labor costs, rather than triggering a sudden wave of defaults. The associated density plot confirms that higher-cost firms exit first, leading to a progressive concentration of the active population around lower labor cost levels.

The smoothness of this exit dynamic is enabled by the elasticity of substitution, which allows firms to partially adjust their input mix in response to rising emission costs, thereby spreading market exits over time. Figure~\ref{1sector.substitution} further illustrates the adjustment mechanisms triggered by the carbon tax. The left panel shows the substitution between emissions and labor: as the carbon price rises, firms decarbonize by reducing their emission intensity and shifting toward labor. Since the elasticity of substitution is sufficiently high, near-complete decarbonization is achieved with a moderate final carbon tax level, while avoiding widespread distress in the sector.

The right panel shows the pass-through effect of carbon costs. The tax is partially transmitted to output prices, which rise more than the emission cost per unit of output, but remain well below the theoretical price level that would prevail in the absence of substitution. In this sense, consumers bear part of the burden of carbon pricing, but the substitution channel substantially mitigates the total price impact.

\begin{figure}
\centerline{\includegraphics[width=0.5\textwidth]{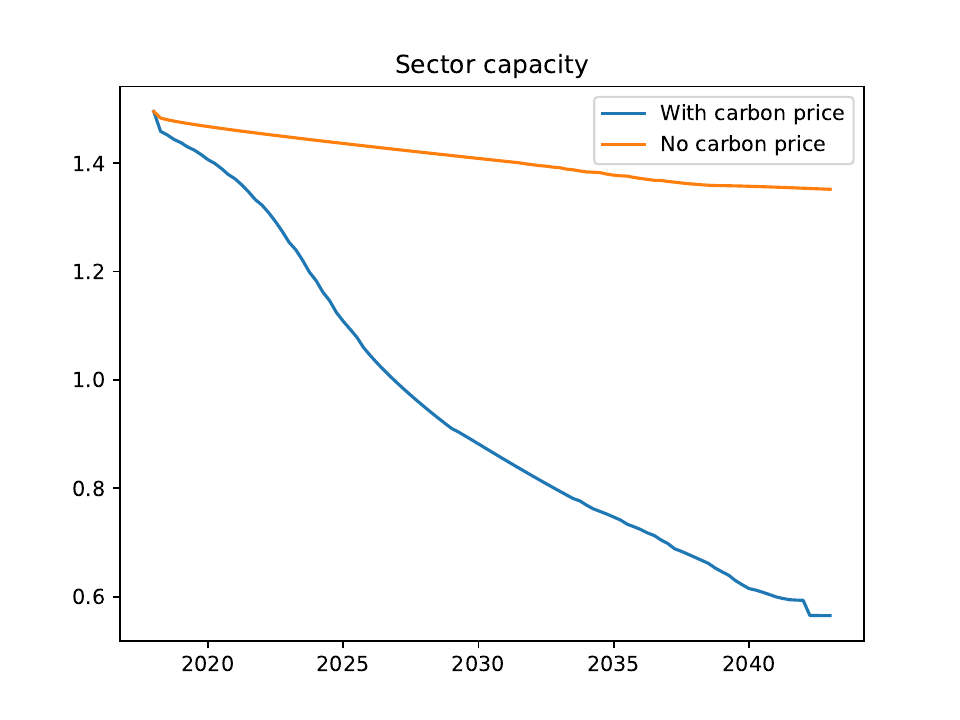}\includegraphics[width=0.5\textwidth]{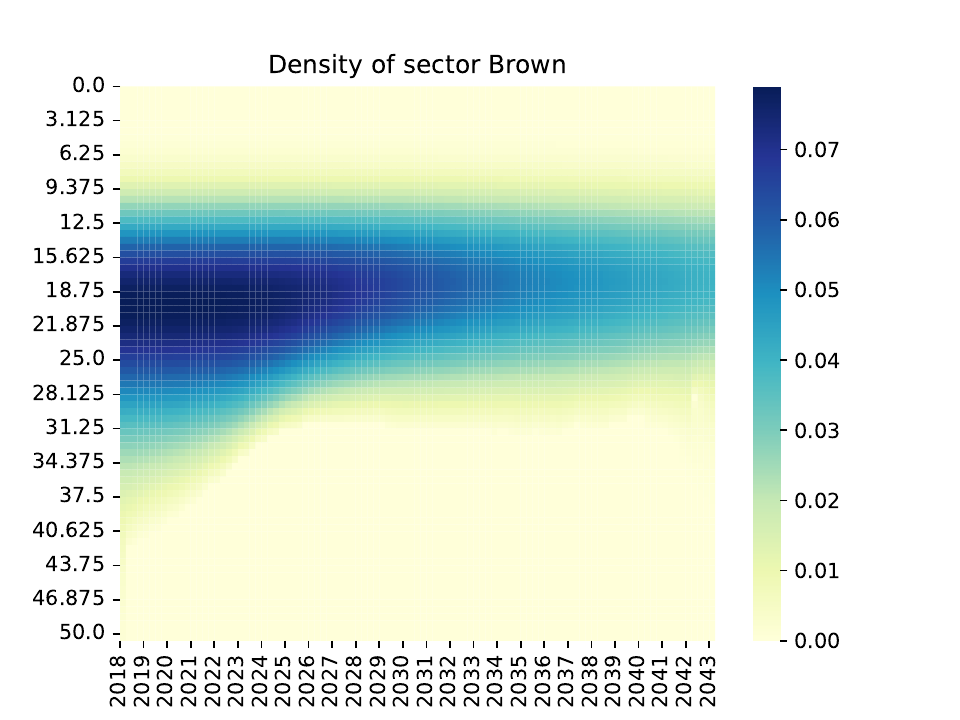}}
\caption{Left: Sector capacity with and without carbon price. Right: Labor cost density under carbon pricing.}
\label{1sector.capacity}
\end{figure}

\begin{figure}
  \centerline{\includegraphics[width=\textwidth]{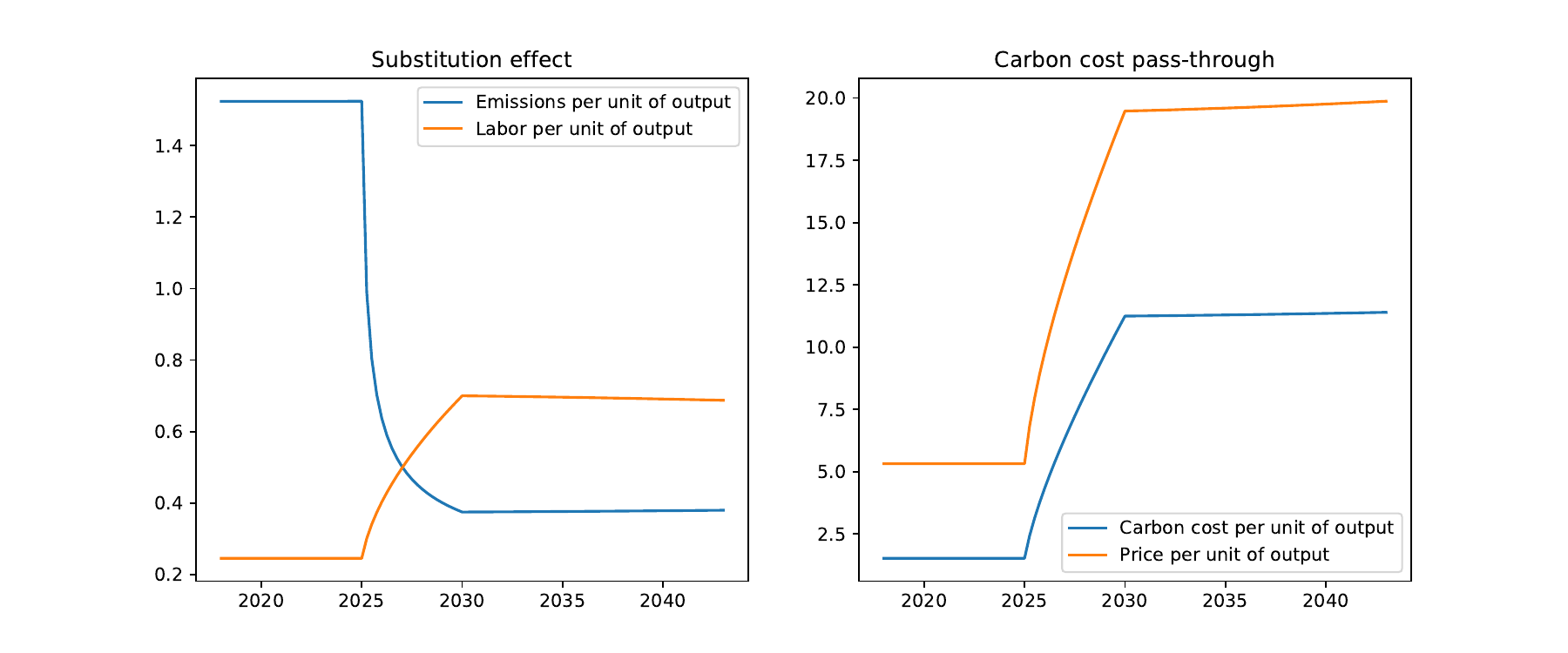}}
\caption{Left: Substitution between emissions and labor. Right: Carbon cost pass-through.}
\label{1sector.substitution}
\end{figure}

\subsection{Example 2: Substitution in a multi-sector economy}

In the three-sector economy, we consider the following input-output structure:   
\begin{itemize}
    \item The Brown sector uses 0.5 units of emissions, 0.4 units of labor, and 0.1 units of Manufacturing output, with a high substitution parameter $\rho_B = 1.5$, reflecting rigid fossil-based technologies.
    \item The Green sector uses 0.01 units of emissions, 0.89 units of labor, and 0.1 units of Manufacturing output, with a moderately high substitution parameter $\rho_G = 0.7$.
    \item The Manufacturing sector uses 0.35 units of Brown output, 0.35 units of Green output, 0.05 units of emissions, and 0.25 units of labor, with a low substitution parameter $\rho_M = 0.1$, allowing substantial input flexibility.
\end{itemize}

We also examine the impact of carbon pricing on final demand across sectors. Each sector faces a power demand function as in Example~\ref{ex:power_demad}, with no exogenous supply. We set $a_B = 100$ and $\epsilon_B = 2$ for the Brown sector, modeling high initial demand for fossil-based technologies and strong sensitivity to price increases. This structure captures the idea that Brown goods are initially dominant but highly exposed to carbon pricing. 

For the Green and Manufacturing sectors, we set $\epsilon_G = \epsilon_M = 0.5$ to reflect inelastic demand: for Green, due to transition policies such as subsidies, and for Manufacturing, due to its central role in the value chain. The scale parameters are set to $a_G = 1$ for Green, reflecting a small initial demand, and $a_M = 10$ for Manufacturing, indicating intermediate baseline demand.

We test the model under a more stringent regulatory scenario, with the carbon price rising from \$1 to \$200 over a ten-year window, from 2025 to 2035.

Figure~\ref{3sectors.emissions} shows the evolution of carbon emissions and labor across the three sectors. Emissions drop sharply in all sectors following the carbon price shock, approaching zero in the Green and Manufacturing sectors. In contrast, emissions in the Brown sector remain strictly positive, reflecting its low elasticity of substitution, which limits the ability to fully decarbonize. Labor in the Brown sector also declines, though less than emissions, indicating that, even under technological rigidity, a substitution effect remains observable, with firms replacing costly emissions with labor to the extent permitted by the production function.

In contrast, both the Green and Manufacturing sectors experience a reallocation of employment rather than a contraction. Their higher substitution elasticities enable them to offset reduced emissions by increasing labor input, thus sustaining production under the carbon constraint. This highlights how decarbonization can shift labor across sectors rather than destroy it, depending on the flexibility of production technologies.

Figure~\ref{3sectors.production} displays the evolution of sectoral production (left panel) and final demand (right panel). All sectors experience a decline in output, though to varying degrees. The Brown sector sees a substantial drop, as it reduces both emissions and labor with limited substitution flexibility. The Green sector experiences only a mild contraction: it successfully replaces emissions with labor, though its dependence on Manufacturing inputs limits a full offset. Manufacturing output declines to a level between that of the Green and Brown sectors, as it partially compensates for the reduction in Brown inputs by increasing reliance on Green goods and labor.

A similar pattern appears in final demand: the Brown sector experiences a near-complete collapse as its highly elastic consumers curtail purchases, whereas demand for Green goods remains almost unchanged and ultimately exceeds that of the Brown sector. 

\begin{figure}[t!]
\centerline{\includegraphics[width=0.9\textwidth]{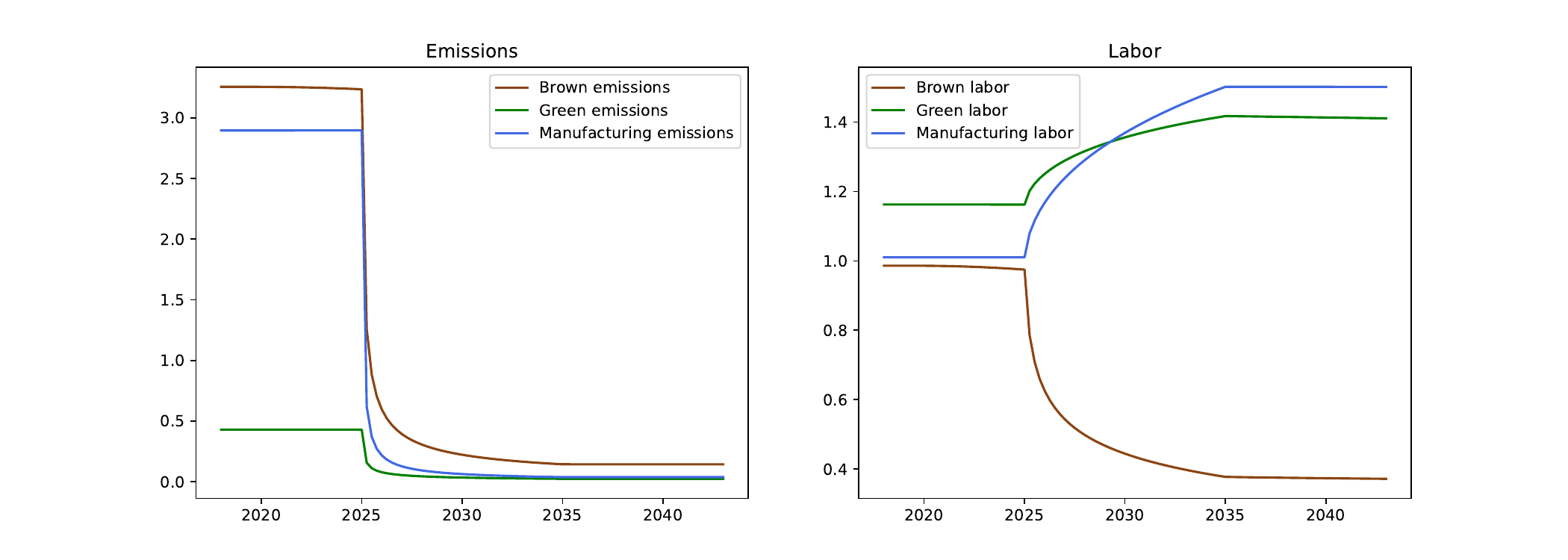}}
\caption{Emissions (left) and labor (right) in the three sectors.}
\label{3sectors.emissions}
\end{figure}

\begin{figure}
\centerline{\includegraphics[width=0.9\textwidth]{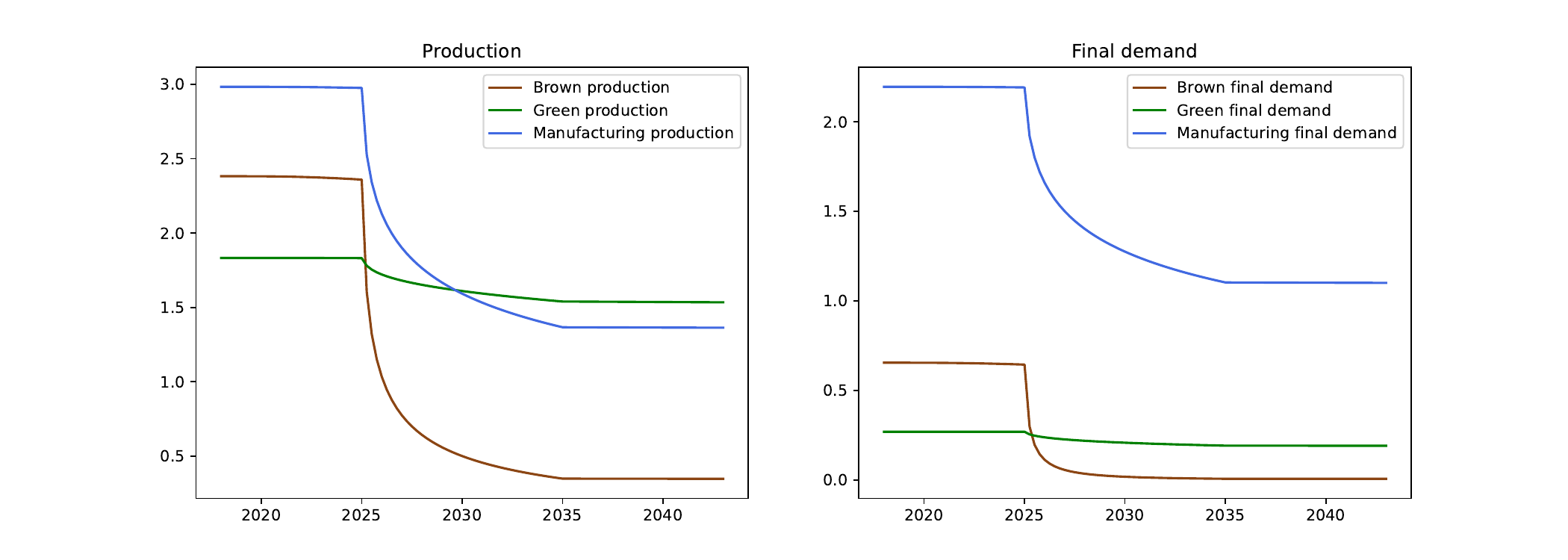}}
\caption{Production (left) and final demand (right) in the three sectors.}
\label{3sectors.production}
\end{figure}

\begin{figure}
\centerline{\includegraphics[width=0.9\textwidth]{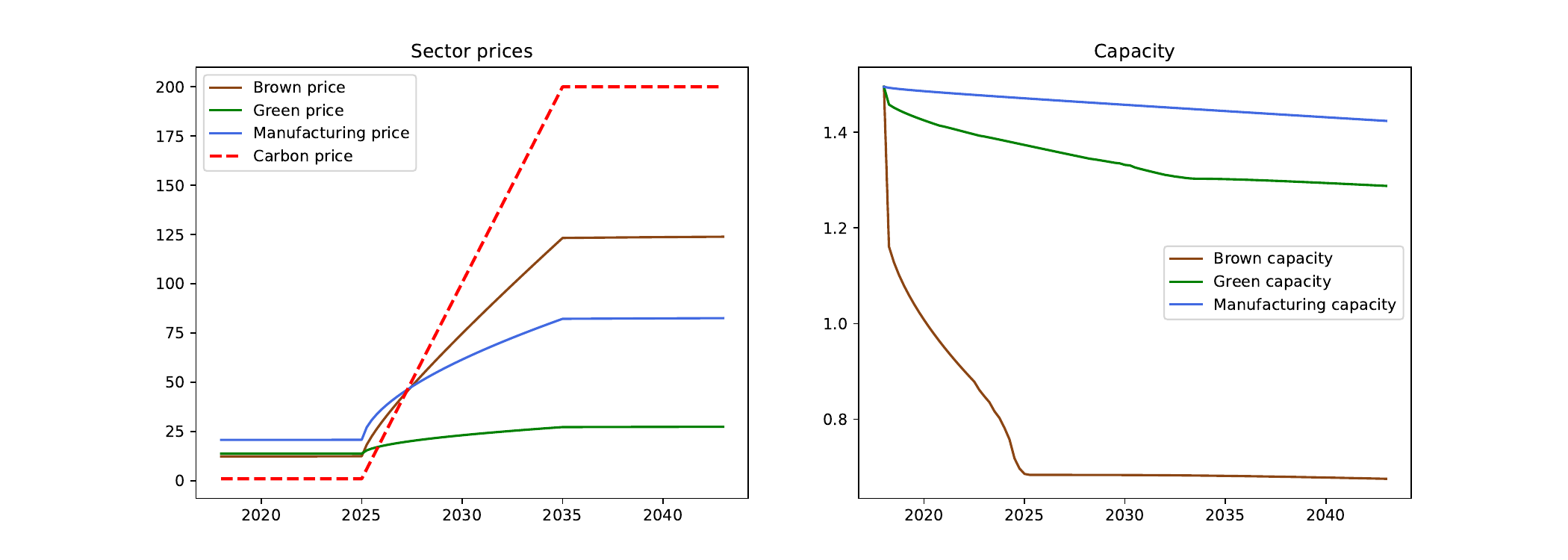}}
\caption{Prices (left) and capacity (right) in the three sectors.}
\label{3sectors.capacity}
\end{figure}

\begin{figure}
\centerline{\includegraphics[width=0.3\textwidth]{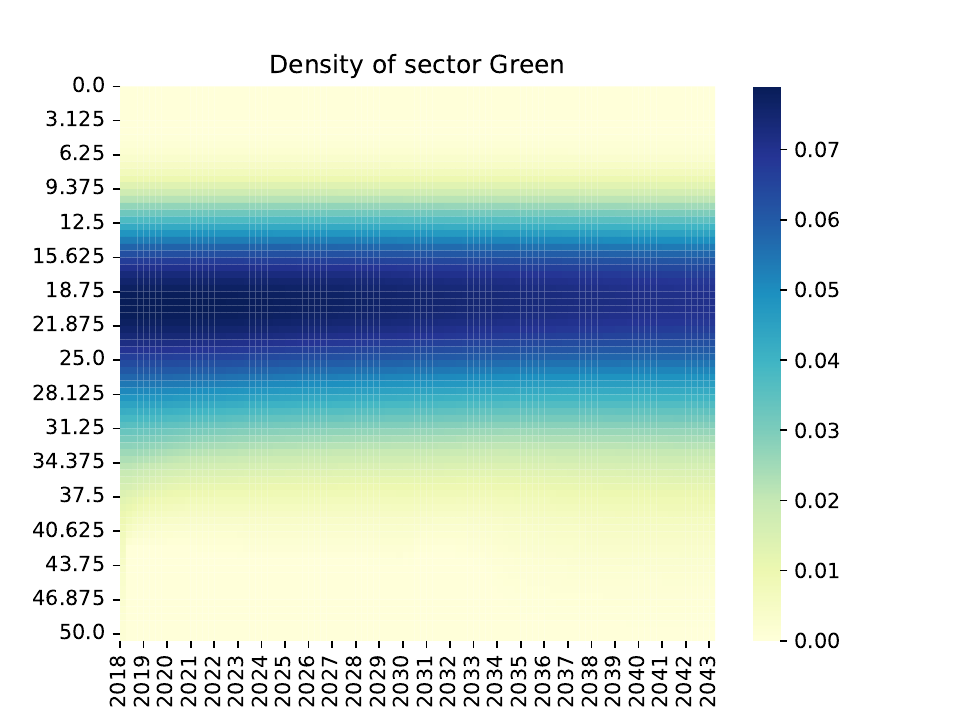}\includegraphics[width=0.3\textwidth]{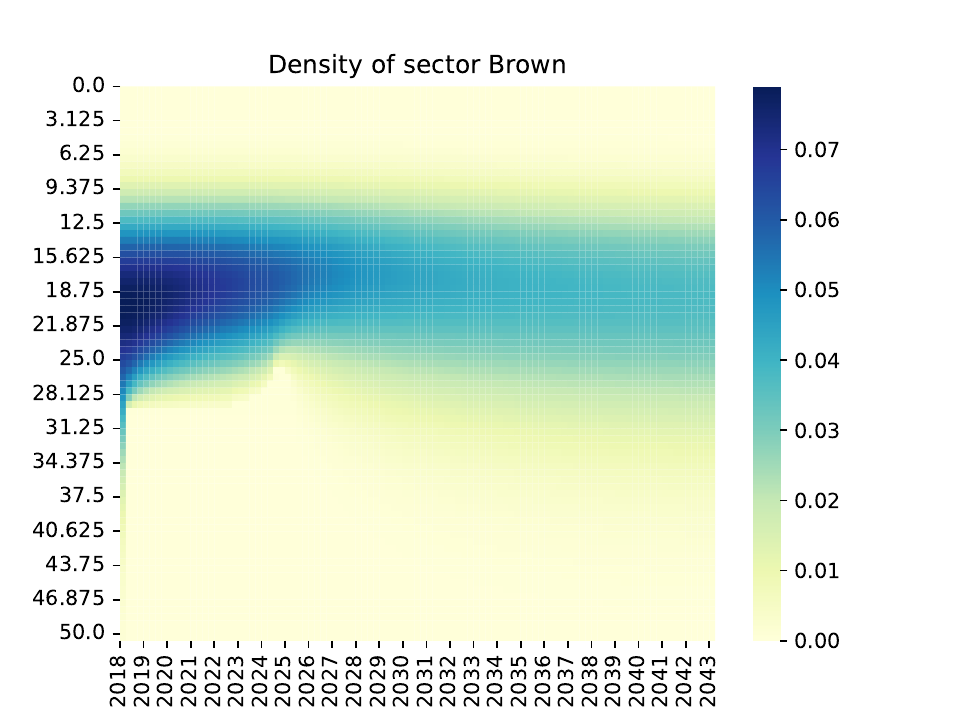}\includegraphics[width=0.3\textwidth]{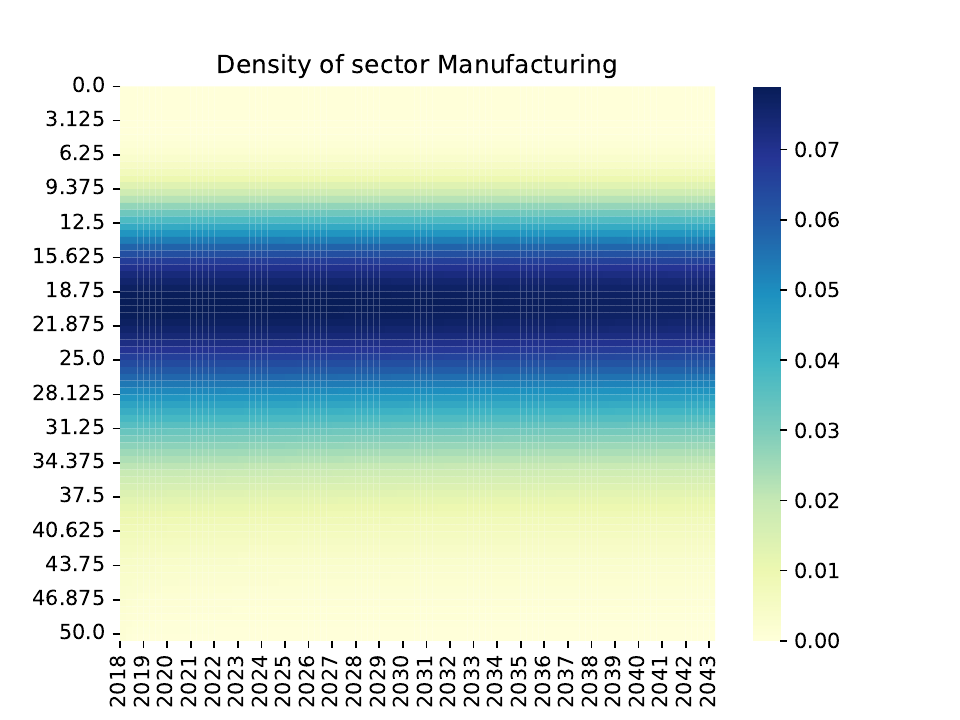}}
\caption{Labor cost densities in the three sectors.}
\label{3sectors.density}
\end{figure}

Finally, Figure~\ref{3sectors.capacity} displays sectoral prices (left panel) and capacity dynamics (right panel), while Figure~\ref{3sectors.density} shows the evolution of labor cost densities. All sectors respond to the carbon price shock by raising their prices, with the Brown sector exhibiting the largest increase.

As in the single-sector model, we observe an anticipation effect in the capacity dynamics of the Brown sector. However, due to its low elasticity of substitution, firms with higher labor costs are forced to exit the market even before the carbon price begins to rise, leading to a sharp initial drop in capacity. This wave of early exits triggers a spillover effect that partially affects the Green sector, which, despite emitting at a much lower rate, remains indirectly exposed through its reliance on Manufacturing inputs. The Manufacturing sector’s capacity remains nearly constant throughout, as it emits very little and increasingly substitutes toward Green input and labor after the shock. Its high elasticity of substitution acts as an effective buffer against the rise in carbon costs.

Notably, when comparing the Brown sector across the single-sector and three-sector models, we observe that long-term capacity stabilizes at a similar level, despite the Brown sector facing a sixfold higher carbon tax and a lower elasticity of substitution in the three-sector setting.

This outcome is explained by differences in the structure of demand. In the single-sector model, Brown firms sell exclusively to final consumers, whose demand is highly elastic. As carbon pricing raises output prices, final demand collapses, leading to substantial revenue losses and widespread firm exit. In contrast, in the three-sector economy, Brown firms are partially sustained by intermediate demand from the Manufacturing sector, even as final demand vanishes. The inter-sector linkage thus not only propagates the carbon cost shock, but also mitigates firm exit by maintaining an additional source of demand.





\section{Conclusions}\label{sec:conclusion}

In this paper, we introduced a novel multi-sector equilibrium model based on mean-field game theory to analyze the propagation of carbon price shocks through the productive economy and along the value chain. The framework accommodates a broad class of realistic production technologies with substitutability between inputs, which plays a key role in mitigating carbon price shocks and shaping decarbonization pathways. Our theoretical analysis not only ensures the existence and uniqueness of equilibria, but also provides a transparent economic interpretation consistent with welfare economics and the centralized social planner problem.

While the model effectively captures agent heterogeneity via idiosyncratic shocks in local input costs (such as labor), it does not account for aggregate sources of uncertainty. A natural extension would be to introduce a common noise component into the mean-field interactions. Other possible generalizations include incorporating stochastic carbon price dynamics and carbon taxation under partial information, for instance when firms respond to uncertain policy trajectories or technological transitions.

From a numerical standpoint, the model could be applied to real-world data, leveraging input–output tables to calibrate production structures and market price data to estimate demand systems.
Such extensions would broaden the model’s empirical relevance and support its application to policy analysis in the context of industrial decarbonization.

\section*{Acknowledgements}

This research was partially supported by Electricité de France. We thank Shiqi Liu, research assistant at CREST, for conducting preliminary computations related to this project. S.P. and T.S. gratefully acknowledge  support from the Carl Zeiss Foundation and  from the FDMAI, the Freiburg Center for Data, Modelling and AI.


\appendix
\section{Proofs of results from Section~\ref{sec:1}}\label{app}

\subsection{Proof of Proposition \ref{prop:profit_max}}
Since $F_i$ is concave and upper semi-continuous, and $W_i$ is convex and lower semi-continuous, the profit function $\Pi_i$ is concave and upper semi-continuous. Additionally, if $F_i$ is strictly concave and $W_i$ is strictly convex, then $\Pi_i$ is strictly concave. In this case, any maximizer, if it exists, is unique.

To establish the existence of an optimal solution, we verify that the profit function satisfies the following coercivity condition: 
\begin{equation}\label{eq:coercivity}
\Pi_i(\mathbf{P},P_E,\mathbf{x}_i)\to-\infty,\quad\text{as}\quad\norm{\mathbf{x}_i}\to\infty, 
\end{equation}
for any price vector $\mathbf{P}\in\R^N_{> 0}$ and any carbon price $P^E\in\R_{> 0}$.

Consider first the case where the production function $F_i$ satisfies Assumption \ref{ass:prod_function} (a). Then, for any $\varepsilon>0$, there exist a constant $C>0$ such that: 
\[
F_i(\mathbf{x}_i)\leq\varepsilon\norm{\mathbf{x}_i}, \quad\text{for all $\mathbf{x}_i$ with $\norm{\mathbf{x}_i}\geq C$}.
\]
Since $W_i$ is strictly increasing and convex, we have:
\[
\liminf_{L\to\infty }\frac{W_i(L)}{L} > 0.
\]
Therefore, there exists a constant $\varepsilon'>0$ (which may depend on both $\mathbf{P}$ and $P^E$) and a constant $C'>0$ such that:
\[
C(\mathbf{P},P_E,\mathbf{x}_i)= \sum^N_{j=1}P_jq_{ij}+P_E E_i+W_i(L_i) \geq \varepsilon'\norm{\mathbf{x}_i}, \quad\text{whenever}\quad \norm{\mathbf{x}_i}\geq C'.
\]
Combining these bounds, the profit function satisfies:
\[
\Pi_i(\mathbf{P},P_E,\mathbf{x}_i)\leq (\varepsilon P_i-\varepsilon')\norm{\mathbf{x}_i},\quad\text{for all $\mathbf{x}_i$ such that $\norm{\mathbf{x}_i}\geq \max\{C,C'\}$}.
\]
Since $\varepsilon$ is arbitrary and $P_i>0$, choosing $\varepsilon\leq\frac{\varepsilon'}{P_i}$ ensures that the right-hand side diverges to $-\infty$ as $\norm{\mathbf{x}_i}\to\infty$, for any $\mathbf{P}\in\R^N_{> 0}$ and any $P_E\in\R_{> 0}$. 

Now, consider the case where the production function $F_i$ satisfies Assumption \ref{ass:prod_function} (b). This allows us to bound the profit function as:
\[
\Pi_i(\mathbf{P},P^E,\mathbf{x}_i)\leq c P_i L_i-\sum^N_{j=1}P_jq_{ij}-P_E E_i-W_i(L_i).
\]
By Assumption \ref{ass:wage_function2}, the right-hand side tends to $-\infty$ as $\norm{\mathbf{x}_i}\to\infty$, for any $\mathbf{P}\in\R^N_{> 0}$ and $P^E\in\R_{> 0}$.

Since the coercivity condition \eqref{eq:coercivity} holds in both cases, the maximization in Problem~\eqref{prob:profit_max} can be restricted to a compact subset of \( \R^{N+2}_{\geq 0} \), and the existence of a maximizer then follows from the extreme value theorem for upper semi-continuous functions.

\subsection{Proof of Proposition \ref{prop:CRS}}
Exploiting the homogeneity property of the production function, the maximized profit can be rewritten as:
\begin{align*}
\overline \Pi_i(\mathbf{P},P_E) 
&= \max_{L_i} \bigg\{ L_i \max_{\mathbf q_i, E_i} \Big[P_i F_i\Big(\frac{\mathbf q_i}{L_i}, \frac{E_i}{L_i},1\Big) - \sum_{j=1}^N P_j \frac{q_{ij}}{L_i} - P_E \frac{E_i}{L_i} \Big] - W_i(L_i) \bigg\} \\
&= \max_{L_i} \bigg\{ L_i \max_{\tilde{\mathbf q}_i, \tilde E_i} \Big[P_i F_i(\tilde{\mathbf q}_i, \tilde E_i,1) - \sum_{j=1}^N P_j \tilde q_{ij} - P_E \tilde E_i \Big] - W_i(L_i) \bigg\} \\
&= W^*_i\bigg( \max_{\tilde{\mathbf q}_i, \tilde E_i} \Big[P_i F_i(\tilde{\mathbf q}_i, \tilde E_i,1) - \sum_{j=1}^N P_j \tilde q_{ij} - P_E \tilde E_i \Big] \bigg) \\
&= W^*_i\big( \widetilde \Pi_i(\mathbf{P},P_E) \big),
\end{align*}
where \(W^*_i\) is the convex conjugate of \(W_i\) introduced in~\eqref{convex conjugate W}, and \(\widetilde \Pi_i(\mathbf{P},P_E)\) is defined as in~\eqref{pitilde}. In the second line, we introduced the normalized variables
\[
\tilde q_{ij} \coloneqq \frac{q_{ij}}{L_i}, \quad \tilde E_i \coloneqq \frac{E_i}{L_i}.
\]

Assuming \(W_i\) is strictly convex and differentiable, the first-order condition for the optimal labor input \(L^*_i\) yields
\[
W'_i(L^*_i) = \widetilde \Pi_i(\mathbf{P},P_E).
\]
Since \(W_i\) is strictly convex, its derivative \(W'_i\) is strictly increasing and hence invertible. Solving for \(L^*_i\) gives
\[
L^*_i = (W'_i)^{-1} \big( \widetilde \Pi_i(\mathbf{P},P_E) \big).
\]

Finally, given the solution \(\tilde{\mathbf{q}}^*_i, \tilde{E}^*_i\) to problem~\eqref{pitilde}, the optimal input quantities and emissions are recovered via the normalized variables:
\[
q^*_{ij} = \tilde q^*_{ij} L^*_i, \quad j=1,\dots,N, \quad E^*_i = \tilde E^*_i L^*_i.
\]
This concludes the proof.

\subsection{Proof of Proposition \ref{ces.ex}}

Let \( F_i \) denote the CES production function with constant returns to scale and positive substitution parameter \( \rho_i > 0 \), given by:
\[
F_i(\mathbf{q}_i, E_i, L_i) = A_i \left( \sum_{j=1}^N \alpha_{ij} q_{ij}^{-\rho_i} + \alpha_{iE} E_i^{-\rho_i} + \alpha_{iL} L_i^{-\rho_i} \right)^{-\frac{1}{\rho_i}}.
\]
Differentiating the normalized profit with respect to the normalized inputs and emissions yields the first-order conditions:
\begin{align*}
&P_i A_i^{-\rho_i}F_i(\tilde{\mathbf{q}}_i, \tilde E_i,1)^{\rho_i+1} \alpha_{ij} \tilde q_{ij}^{-\rho_i-1}-P_j=0,\quad \quad j=1,\dots,N,\\
&P_i A_i^{-\rho_i} F_i(\tilde{\mathbf{q}}_i, \tilde E_i,1)^{\rho_i+1} \alpha_{iE} \tilde E_i^{-\rho_i-1}-P_E=0.
\end{align*}

The left-hand side of each equation is a decreasing function of the corresponding input (i.e., \( \tilde q_{ij} \) or \( \tilde E_i \)), holding all other parameters fixed. Moreover, for \( \rho_i > 0 \), the following asymptotic limits hold:
\begin{align*}
&\lim_{\tilde q_{ij} \to 0} A_i^{-\rho_i} F_i(\tilde{\mathbf{q}}_i, \tilde E_i,1)^{\rho_i+1} \alpha_{ij} \tilde q_{ij}^{-\rho_i-1} = \alpha_{ij}^{-1/\rho_i} A_i,\\
&\lim_{\tilde q_{ij} \to \infty} A_i^{-\rho_i} F_i(\tilde{\mathbf{q}}_i, \tilde E_i,1)^{\rho_i+1} \alpha_{ij} \tilde q_{ij}^{-\rho_i-1} = 0,\\
&\lim_{\tilde E_i \to 0} A_i^{-\rho_i} F_i(\tilde{\mathbf{q}}_i, \tilde E_i,1)^{\rho_i+1} \alpha_{iE} \tilde E_i^{-\rho_i-1} = \alpha_{iE}^{-1/\rho_i} A_i,\\
&\lim_{\tilde E_i \to \infty} A_i^{-\rho_i} F_i(\tilde{\mathbf{q}}_i, \tilde E_i,1)^{\rho_i+1} \alpha_{iE} \tilde E_i^{-\rho_i-1} = 0.
\end{align*}

It follows that the optimality conditions are:
\begin{align*}
P_i A_i^{-\rho_i} F_i(\tilde{\mathbf{q}}_i, \tilde E_i,1)^{\rho_i+1} \alpha_{ij} \tilde q_{ij}^{-\rho_i-1} = P_j \quad \text{if} \quad \frac{P_j}{P_i} \leq \alpha_{ij}^{-1/\rho_i} A_i, \quad \text{and} \quad \tilde q_{ij} = 0 \text{ otherwise},
\end{align*}
for \( j=1,\dots,N \), and
\begin{align*}
P_i A_i^{-\rho_i} F_i(\tilde{\mathbf{q}}_i, \tilde E_i,1)^{\rho_i+1} \alpha_{iE} \tilde E_i^{-\rho_i-1} = P_E \quad \text{if} \quad \frac{P_E}{P_i} \leq \alpha_{iE}^{-1/\rho_i} A_i, \quad \text{and} \quad \tilde E_i = 0 \text{ otherwise}.
\end{align*}

It is easy to verify that
\begin{align*}
&\lim_{\tilde q_{ij} \to 0^+} F_i(\tilde{\mathbf{q}}_i, \tilde E_i,1) = 0 \quad \text{for some } j=1,\dots,N, \\
&\lim_{\tilde E_i \to 0^+} F_i(\tilde{\mathbf{q}}_i, \tilde E_i,1) = 0,
\end{align*}
so that if \( \tilde q_{ij} = 0 \) for any \( j \), or \( \tilde E_i = 0 \), then the total output is zero. In such a scenario, all inputs vanish: \( \tilde q_{ij} = 0 \) for all \( j \), and \( \tilde E_i = 0 \).

The optimality conditions can therefore be summarized as:
\[
\text{if} \quad \frac{P_j}{P_i} \leq \alpha_{ij}^{-1/\rho_i} A_i \quad \text{for all } j=1,\dots,N, \quad \text{and} \quad \frac{P_E}{P_i} \leq \alpha_{iE}^{-1/\rho_i} A_i,
\]
then:
\[
P_i A_i^{-\rho_i} F_i(\tilde{\mathbf{q}}_i, \tilde E_i,1)^{\rho_i+1} \alpha_{ij} \tilde q_{ij}^{-\rho_i-1} = P_j, \quad \text{for all } j=1,\dots,N,
\]
and:
\[
P_i A_i^{-\rho_i} F_i(\tilde{\mathbf{q}}_i, \tilde E_i,1)^{\rho_i+1} \alpha_{iE} \tilde E_i^{-\rho_i-1} = P_E.
\]
Otherwise, we have \( \tilde q_{ij} = 0 \) for all \( j \), and \( \tilde E_i = 0 \).

To obtain the explicit expressions, suppose the above condition is satisfied, and denote by \( \tilde Q^*_i := F_i(\tilde{\mathbf{q}}^*_i, \tilde E^*_i,1) \) the optimal production normalized by labor. Then the optimal inputs are:
\begin{align}\label{optquant}
\tilde q^*_{ij} = \tilde Q^*_i \left( \frac{P_j A_i^{\rho_i}}{P_i \alpha_{ij}} \right)^{-\frac{1}{\rho_i+1}}, \quad
\tilde E^*_i = \tilde Q^*_i \left( \frac{P_E A_i^{\rho_i}}{P_i \alpha_{iE}} \right)^{-\frac{1}{\rho_i+1}}.
\end{align}

Substituting these into \( F_i(\tilde{\mathbf{q}}^*_i, \tilde E^*_i,1) \), we obtain the optimal normalized output:
\begin{align}\label{optprod}
\tilde Q^*_i = A_i \alpha_{iL}^{-\frac{1}{\rho_i}} \left( 1 - \sum_{j=1}^N \left( \frac{P_j \alpha_{ij}^{1/\rho_i}}{A_i P_i} \right)^{\frac{\rho_i}{\rho_i+1}} - \left( \frac{P_E \alpha_{iE}^{1/\rho_i}}{A_i P_i} \right)^{\frac{\rho_i}{\rho_i+1}} \right)^{\frac{1}{\rho_i}}.
\end{align}
This expression is valid whenever the term inside the parentheses is non-negative; otherwise, the optimal production is zero.

Thus, the optimal production and inputs are given by \eqref{optprod} and \eqref{optquant}, respectively, if and only if
\[
\sum_{j=1}^N \left( \frac{P_j \alpha_{ij}^{1/\rho_i}}{A_i P_i} \right)^{\frac{\rho_i}{\rho_i+1}} + \left( \frac{P_E \alpha_{iE}^{1/\rho_i}}{A_i P_i} \right)^{\frac{\rho_i}{\rho_i+1}} \leq 1;
\]
otherwise, both production and input levels are zero.

Under the condition of strictly positive production, the normalized profit function \( \widetilde \Pi_i \) is given by:
\begin{align*}
\widetilde \Pi_i(\mathbf{P},P_E)
&= \tilde Q^*_i \left[ P_i - \sum_{j=1}^N P_j \left( \frac{P_j A_i^{\rho_i}}{P_i \alpha_{ij}} \right)^{-\frac{1}{\rho_i+1}} - P_E \left( \frac{P_E A_i^{\rho_i}}{P_i \alpha_{iE}} \right)^{-\frac{1}{\rho_i+1}} \right] \\
&= P_i A_i \alpha_{iL}^{-1/\rho_i} \left( 1 - \sum_{j=1}^N \left( \frac{P_j \alpha_{ij}^{1/\rho_i}}{A_i P_i} \right)^{\frac{\rho_i}{\rho_i+1}} - \left( \frac{P_E \alpha_{iE}^{1/\rho_i}}{A_i P_i} \right)^{\frac{\rho_i}{\rho_i+1}} \right)^{\frac{\rho_i+1}{\rho_i}}.
\end{align*}

This concludes the proof.

\section{Proof of Lemma \ref{lem:saddle_point}}\label{app:1}
Let $(\mathbf{P}^*,\cQ^*,\mathbf{E}^*,\mathbf{L}^*)\in\R^N_{> 0}\times \R^{N(N+2)}_{\geq 0}$ be a saddle point of the Lagrangian defined by \eqref{eq:lagrangian}. From the left inequality in Definition \ref{def:saddle_point}, it follows that, for each $i=1,\dots,N$, the tuple $(\mathbf{q}^*_i,E^*_i,L^*_i)$ solves the profit maximization problem \eqref{prob:profit_max}.

Next, from the right inequality in Definition \ref{def:saddle_point}, we have, for all $\mathbf{P} \in \mathbb R^N_{\geq 0}$, 
$$
\cL(\mathbf{P}^*,\cQ^*,\mathbf{E}^*,\mathbf{L}^*)\leq\cL(\mathbf{P},\cQ^*,\mathbf{E}^*,\mathbf{L}^*).
$$
Since, by Assumption \ref{ass:demand}, the Lagrangian is differentiable with respect to $\mathbf{P}$ on $\R^N_{>0}$, it follows that:
$$
\frac{\partial\cL}{\partial P_i}(\mathbf{P}^*,\cQ^*,\mathbf{E}^*,\mathbf{L}^*) = F_i(\mathbf{q}^*_i, E^*_i,L^*_i)-\sum^N_{j=1}q^*_{ji}-D_i(P^*_i) = 0,\quad i=1,\dots,N,
$$
which corresponds to the clearing condition \eqref{eq:clearing}.

Conversely, assume $(\mathbf{P}^*,\cQ^*,\mathbf{E}^*,\mathbf{L}^*)\in\R^N_{> 0}\times \R^{N(N+2)}_{\geq 0}$ is a CE. By condition (ii) of Definition \ref{def:Nash_equilibrium}, the left inequality in Definition \ref{def:saddle_point} is satisfied. The right inequality follows from the differentiability and convexity of the Lagrangian with respect to $\mathbf{P}\in\R^N_{>0}$, combined with the fact that its derivative at the saddle point is zero due to the market clearing condition.

\section{Proofs of technical lemmas from Section~\ref{sec:3}}\label{app:2}

\subsection{Proof of Lemma \ref{lem:compactness}}

Since $\cV_p(\overline{\cO}_i) \times \cM_p([0,T]\times\overline{\cO}_i)$, endowed with the product topology $\tau_p \otimes \tau_p$, is metrizable, it suffices to show that $\cR^i(m_0)$ is sequentially compact. To this end, consider a sequence $(m^n, \mu^n)_{n \geq 1} \subset \cR^i(m_0)$. Define the test function
\[
u_k(t,x) \coloneqq (T+1 - t)\, \phi_k(x), \quad k \geq 1,
\]
where
\[
\phi_k(x) \coloneqq \left(1 + x^q f\left(\frac{x}{k}\right)\right) \Ind_{x \leq k} + \left(1 + k^3 f(1)\right) \Ind_{x > k}, \quad q \leq p',
\]
and the function $f$ is given by
\[
f(z) \coloneqq \frac{q^2 + q}{q^2 + 3q + 2} z^2 - \frac{2q^2 + 4q}{q^2 + 3q + 2} z + 1.
\]
It is straightforward to verify that $f$ is smooth and positive on $[0,1]$, with $f(0) = 1$, ensuring that $\phi_k \in \cC^2_b(\cO_i)$ and is strictly positive. In particular, the first and second derivatives of $\phi_k$ are given by:
\begin{align*}
\phi'_k(x) &= C_q \left(1 - \frac{x}{k}\right)^2 x^{q-1} \Ind_{x \leq k}, \\
\phi''_k(x) &= C_q \left[ -\frac{2}{k}\left(1 - \frac{x}{k}\right) x^{q-1} + (q - 1) \left(1 - \frac{x}{k}\right)^2 x^{q-2} \right] \Ind_{x \leq k},
\end{align*}
where 
\[
C_q = \frac{q(q+1)(q+2)}{q^2 + 3q + 2}.
\]

Substituting this test function into Equation \eqref{eq:weak_FP}, and noting that the left-hand side is non-negative, we obtain
\[
\int_{\cO_i} u_k(0,x)\, m_0(dx) + \int_0^T \int_{\overline\cO_i}
\left( \frac{\partial u_k}{\partial t} + \crL_i u_k \right)(t,x)\, m^n_t(dx)dt \geq 0.
\]
This implies the estimate:
\begin{align*}
\int_0^T \int_{\overline\cO_i}
\phi_k(x)\, m^n_t(dx)dt &\leq (T+1) \int_{\cO_i} \phi_k(x)\, m_0(dx) \\
& + \int_0^T \int_{\overline\cO_i} (T+1 - t) \left[ \alpha_i(t,x)\, \phi'_k(x) + \frac{\sigma_i^2(t,x)}{2}\, \phi''_k(x) \right]\,m^n_t(dx) dt.
\end{align*}

Taking $q = 0$ so that $\phi_k \equiv 1$ yields the bound:
\[
\int_0^T \int_{\overline\cO_i}\,m^n_t(dx) dt \leq (T+1) \int_{\cO_i} m_0(dx).
\]

Next, under Assumption~\ref{ass:SDE}, and using the explicit expressions for $\phi'_k$ and $\phi''_k$, we can bound the generator term as:
\begin{align*}
\left| \alpha_i(t,x)\, \phi'_k(x) \right| + \frac{1}{2} \sigma_i^2(t,x) \left| \phi''_k(x) \right| 
\leq C_q \left( x^{q-1} + \frac{1 + k^\beta}{k} x^{q-1} + (q-1) \frac{1 + k^\beta}{k} x^{q-2} \right) \Ind_{x \leq k}.
\end{align*}
Substituting this bound into the inequality above, and recalling that $\beta \in [0,1]$, we obtain:
\begin{align*}
\int_0^T \int_{\overline\cO_i}
\phi_k(x)\, m^n_t(dx)dt &\leq (T+1) \int_{\cO_i} \phi_k(x)\, m_0(dx) \\
&\quad + C \int_0^T \int_{\overline\cO_i}
\left( x^{q-1} + (q-1)x^{q-2} \right) \Ind_{x \leq k}\, m^n_t(dx)dt,
\end{align*}
for some constant $C > 0$ independent of $n$.

Applying the monotone convergence theorem as $k \to \infty$, we deduce:
\begin{align*}
\int_0^T \int_{\overline\cO_i}
(1 + x^q)\, m^n_t(dx)dt &\leq (T+1) \int_{\cO_i} (1 + x^q)\, m_0(dx) \\
&\quad + C \int_0^T \int_{\overline\cO_i}
\left( x^{q-1} + (q-1) x^{q-2} \right)\, m^n_t(dx)dt.
\end{align*}
Applying this estimate iteratively for $q = 1, \dots, \lceil p' \rceil$, we conclude that
\[
\sup_{n \geq 1} \int_0^T \int_{\overline\cO_i}
(1 + x^{p'})\, m^n_t(dx)dt < \infty,
\]
which proves that the sequence $(m^n)_{n \geq 1}$ is relatively compact in the topology $\tau_p$; see \cite[Corollary A.4]{dumitrescu2023linear}.

A similar argument using the same test function $u_k$ yields the relative compactness of the sequence $(\mu^n)_{n \geq 1}$.

Finally, the argument in \cite[Theorem 2.13]{Dumitrescu2021LP} shows that the weak limit of a convergent subsequence $(m^n, \mu^n)$ satisfies Equation~\eqref{eq:weak_FP}. Since $p \geq 1$ and the generator $\crL_i$ has at most linear growth, we conclude that the limit pair $(m,\mu)$ belongs to $\cR^i(m_0)$. This establishes the compactness of $\cR^i(m_0)$.

\subsection{Proof of Lemma \ref{lem:MFG_saddle_point}}

Let $(\mathbf{P}^*,\mathbf{m}^*,\boldsymbol{\mu}^*) \in \cP^+_q\times\cR(\mathbf{m}_0)$ be a saddle point of the Lagrangian defined in \eqref{eq:MFG_lagrangian}. By the definition of a saddle point, we have  
\[
\cL(\mathbf{P}^*,\mathbf{m},\boldsymbol{\mu}) \leq \cL(\mathbf{P}^*,\mathbf{m}^*,\boldsymbol{\mu}^*) \quad \text{for all } (\mathbf{m},\boldsymbol{\mu})\in\cV^N_p\times\cV^N_p.
\]
This directly establishes Property (iii) of Definition \ref{def:MFG_equilibrium}. Moreover, Property (ii) follows immediately from the assumptions of the lemma.  

To establish Property (i), we recall that the Minimax equality \eqref{eq:MFG_minimax} ensures  
\[
\cL(\mathbf{P}^*,\mathbf{m}^*,\boldsymbol{\mu}^*) \leq \cL(\mathbf{P},\mathbf{m}^*,\boldsymbol{\mu}^*) \quad \text{for all } \mathbf{P} \in L^q.
\]
As a consequence, we obtain that for $t$-almost every $t\in[0,T]$ and for all $\mathbf{P} \in L^q$,  
\begin{align*}
\sum_{j=1}^N\int_{\overline\cO_j} \overline\Pi_j(\mathbf{P}^*(t),x)\,m^{j,*}_t(dx)-\Delta_j(t,P^{*}_j(t)) 
\leq \sum_{j=1}^N \int_{\overline\cO_j} \overline\Pi_j(\mathbf{P}(t),x)\,m^{j,*}_t(dx)-\Delta_j(t,P_j(t)).
\end{align*}
For notational simplicity, we suppress the explicit dependence of $\overline\Pi_j$ on $P_E(t)$.

Since $\Delta_j(t,\cdot)$ is differentiable on $\mathbb{R}_{>0}$ for all $t\in [0,T]$, as ensured by Assumption \ref{ass:demand_uniform}, the first-order optimality condition gives  
\begin{equation}\label{eq:FOC_1}
\frac{\partial}{\partial P_i} \sum_{j=1}^N \int_{\overline\cO_j} \overline\Pi_j(\mathbf{P}^*(t),x)\,m^{j,*}_t(dx) - D_i(t,P^*_i(t)) = 0.
\end{equation}
To justify the interchange of the derivative and the integral in the first term on the left-hand side of \eqref{eq:FOC_1}, we express the partial derivative as a limit, and applying the linearity of the integral, we obtain:
\begin{align}\label{eq:FOC_limit}
\frac{\partial}{\partial P_i} \int_{\overline\cO_j} \overline\Pi_j(\mathbf{P}^*(t),x)\,m^{j,*}_t(dx)  
= \lim_{h\to 0} \int_{\overline\cO_j} \frac{\overline\Pi_j(\mathbf{P}^*(t) + h\mathbf{e}_i, x) - \overline\Pi_j(\mathbf{P}^*(t), x)}{h} \,m^{j,*}_t(dx),
\end{align}  
where \( \mathbf{e}_i \) denotes the unit vector in \( \mathbb{R}^N \) corresponding to the \( i \)-th component.  

Using the explicit expression for the maximized instantaneous profit from Proposition \ref{prop:max_inst_profit}, we obtain:  
\begin{align*}
\left| \overline\Pi_j(\mathbf{P}^*(t) + h\mathbf{e}_i, x) - \overline\Pi_j(\mathbf{P}^*(t), x) \right|  
&= \left| W^*_j\big(\widetilde\Pi_j(\mathbf{P}^*(t) + h\mathbf{e}_i) - x\big) - W^*_j\big(\widetilde\Pi_j(\mathbf{P}^*(t)) - x\big) \right| \\  
&\leq \left| W^*_j\big(\widetilde\Pi_j(\mathbf{P}^*(t)) + h\mathbf{e}_i\big) - W^*_j\big(\widetilde\Pi_j(\mathbf{P}^*(t))\big) \right|,
\end{align*}  
where the inequality follows from the convexity of \( W^*_j \).  

To establish an integrable bound for the expression inside the integral on the right-hand side of Equation \eqref{eq:FOC_limit}, it suffices to show that the limit 
\begin{equation}\label{conv_limit}
\lim_{h\to 0} \frac{W^*_j\big(\widetilde\Pi_j(\mathbf{P}^*(t)) + h\mathbf{e}_i\big) - W^*_j\big(\widetilde\Pi_j(\mathbf{P}^*(t))\big)}{h}
\end{equation}
exists and is finite.

By the strict concavity of \( F_i \), there exists a unique pair \( (\tilde{\mathbf{q}}^*_i, \tilde{E}^*_i) \) solving the optimization problem \eqref{pitilde} for each \( i=1,\dots,N \). Consequently, by Danskin's theorem (see \cite[Proposition B.25]{bertsekas1999nonlinear}), the function \( \widetilde \Pi_j(\mathbf{P}) \) is differentiable with respect to \( P_i \).  

Furthermore, since \( W_i \) is strictly convex and differentiable, its convex conjugate \( W^*_i \) is also differentiable in the interior of its domain (see \cite[Theorem 26.3]{rockafellar1970convex}). This ensures that the limit in \eqref{conv_limit} exists and is finite.  

Next, applying Lebesgue's dominated convergence theorem, we can interchange the integral and the limit on the right-hand side of Equation \eqref{eq:FOC_limit}, thereby validating differentiation under the integral sign.  

As a result, the first-order condition \eqref{eq:FOC_1} can be rewritten as:  
\begin{align}\label{eq:FOC_2}
\sum_{j=1}^N\int_{\overline\cO_j} \frac{\partial}{\partial P_i}\overline\Pi_j(\mathbf{P}^*(t),x)\,m^{j,*}_t(dx) - D_i(t,P^{*}_i(t)) = 0.
\end{align}

Applying Danskin's theorem once again, we obtain an explicit expression for the derivative of \( \overline\Pi_j \) with respect to \( P_i \):
\[
\frac{\partial}{\partial P_i} \overline\Pi_j(\mathbf{P}^*(t), x) = \Ind_{i=j} F_j(\mathbf{q}^{*}_j(\mathbf{P}^*(t), x), E^{*}_j(\mathbf{P}^*(t), x), L^{*}_j(\mathbf{P}^*(t), x)) - q^*_{ji}(\mathbf{P}^*(t), x),
\]
where \( (\mathbf{q}^{*}_j, E^{*}_j, L^{*}_j)(\mathbf{P}^*(t), x) \) denotes the optimal input allocation solving the instantaneous profit maximization problem \eqref{eq:max_inst_profit} at price level \( \mathbf{P}^*(t) \), for a given \( (t, x) \in [0, T] \times \cO_j \).

Substituting this expression into Equation \eqref{eq:FOC_2} and summing over all terms, we recover the market-clearing conditions. This confirms that \( \mathbf{P}^* \) satisfies Property (i) of Definition \ref{def:MFG_equilibrium}, thereby completing the proof.

\end{document}